\newif\ifdraft
\draftfalse

\newif\ifinc
\incfalse

\newif\ifans
\ansfalse

\newif\ifsample
\samplefalse

\newif\iffull

\fulltrue 

\providecommand{\remove}[1]{}

\iffull
	\documentclass[11pt,letterpaper]{article}
\else
	\documentclass[letter,envcountsame,envcountsect]{llncs}
    
\fi

\iffull
 	\usepackage{fullpage}
 	\usepackage[margin=4mm, small, labelfont=bf]{caption}
	\ifdraft 
		\usepackage[draft,notes=true,later=false]{dtrt} 
	\else
		\usepackage[notes=true,later=false]{dtrt} 
    \fi
 
\usepackage[natbib=true,backend=bibtex,backref=true,style=alphabetic,maxnames=2,maxalphanames=2,maxbibnames=99, sorting=anyt]{biblatex}
\else
\usepackage[notes=xxx,llncssubsub]{dtrt}

\usepackage[style=alphabetic,firstinits,backend=bibtex,maxalphanames=2,minalphanames=2,maxbibnames=6,natbib=true]{biblatex}

\fi

\bibliography{crypto}

\iffull
\usepackage{titlesec}
\usepackage{amsthm}
\else
\let\proof\relax
\let\endproof\relax
\usepackage{amsthm}
\fi

\remove{
\usepackage[
pagebackref,backref=true,
hyperfootnotes=false,
colorlinks=true,
urlcolor=externallinkcolor,
linkcolor=internallinkcolor,
filecolor=externallinkcolor,
citecolor=internallinkcolor,
breaklinks=true,
pdfstartview=FitH,
pdfpagelayout=OneColumn]{hyperref}
}

\usepackage{mathtools}
\usepackage{bbm}
\usepackage{amsthm,amsfonts,amsmath,amssymb,color,amsthm,boxedminipage,url,xparse}
\definecolor{internallinkcolor}{rgb}{0,.5,0}
\usepackage{amsthm}

\usepackage{xspace}
\usepackage{nicefrac}

\usepackage{cleveref}

\usepackage{amssymb}
\usepackage{amsmath}
\usepackage{bbm}
\usepackage{paralist}
\usepackage{enumitem}

\setlist[itemize]{leftmargin=*}
\setlist[enumerate]{leftmargin=*}

\usepackage{mathtools}
\usepackage{paralist}

\providecommand{\remove}[1]{}

\ifdraft
    \newcommand{\authnote}[2]{{\bf [{\color{red} #1's Note:} {\color{blue} #2}]}}
\else
    \newcommand{\authnote}[2]{}
\fi

\remove{
\titleclass{\subsubsubsection}{straight}[\subsection]

\newcounter{subsubsubsection}[subsubsection]
\renewcommand\thesubsubsubsection{\thesubsubsection.\arabic{subsubsubsection}}

\titleformat{\subsubsubsection}
{\normalfont\normalsize\bfseries}{\thesubsubsubsection}{1em}{}
\titlespacing*{\subsubsubsection}
{0pt}{3.25ex plus 1ex minus .2ex}{1.5ex plus .2ex}

\makeatletter
\renewcommand\paragraph{\@startsection{paragraph}{5}{\z@}%
	{3.25ex \@plus1ex \@minus.2ex}%
	{-1em}%
	{\normalfont\normalsize\bfseries}}
\renewcommand\subparagraph{\@startsection{subparagraph}{6}{\parindent}%
	{3.25ex \@plus1ex \@minus .2ex}%
	{-1em}%
	{\normalfont\normalsize\bfseries}}
\def\toclevel@subsubsubsection{4}
\def\toclevel@paragraph{5}
\def\toclevel@paragraph{6}
\def\l@subsubsubsection{\@dottedtocline{4}{7em}{4em}}
\def\l@paragraph{\@dottedtocline{5}{10em}{5em}}
\def\l@subparagraph{\@dottedtocline{6}{14em}{6em}}
\makeatother

\setcounter{secnumdepth}{4}
\setcounter{tocdepth}{4}
}

\let\originalleft\left
\let\originalright\right
\renewcommand{\left}{\mathopen{}\mathclose\bgroup\originalleft}
\renewcommand{\right}{\aftergroup\egroup\originalright}

\newcommand{\aka} {also known as,\xspace}
\newcommand{\resp}{resp.,\xspace}
\newcommand{\ie}  {i.e.,\xspace}
\newcommand{\eg}  {e.g.,\xspace}

\newcommand{\wrt} {with respect to\xspace}
\newcommand{\wlg} {without loss of generality\xspace}
\newcommand{\cf}{cf.,\xspace}

\newcommand{\ceil}[1]{\left\lceil #1 \right\rceil}
\newcommand{\ip}[1]{\iprod{#1}}
\newcommand{\iprod}[1]{\langle #1 \rangle}

\newcommand{\set}[1]{\left\{#1\right\}}

\newcommand{\paren}[1]{\left(#1\right)}

\newcommand{\floor}[1]{\left \lfloor#1 \right \rfloor}
\newcommand{\eqdef}{:=}

\newcommand{\N}{{\mathbb{N}}}

\newcommand{\zo}{\set{0,1}}
\newcommand{\oo}{\mo}

\newcommand{\mo}{\set{\!\shortminus 1,1\!}}

\newcommand{\zn}{{\zo^n}}
\newcommand{\zs}{{\zo^\ast}}

\newcommand{\condition}{\;\ifnum\currentgrouptype=16 \middle\fi|\;}

\newcommand{\xor}{\oplus}

\newcommand{\eps}{\varepsilon}

\newcommand{\from}{\leftarrow}
\newcommand{\la}{\gets}

\makeatletter
\@ifpackageloaded{complexity}{}{%
\usepackage[classfont=roman,langfont=roman,funcfont=roman]{complexity}}
\makeatother

\newcommand{\Dec}{\MathAlgX{Dec}}

\newcommand{\negl}{\operatorname{neg}}

\newcommand{\Supp}{\operatorname{Supp}}
\newcommand{\GL}{\operatorname{GL}}

\newcommand{\MathAlg}[1]{\mathsf{#1}}

\newcommand{\Ensuremath}[1]{\ensuremath{#1}\xspace}

\newcommand{\MathAlgX}[1]{\Ensuremath{\MathAlg{#1}}}

\renewcommand{\cref}{\Cref}
\usepackage{aliascnt}

\iffull

\newtheorem{theorem}{Theorem}[section]

\newaliascnt{lemma}{theorem}
\newtheorem{lemma}[lemma]{Lemma}
\aliascntresetthe{lemma}
\crefname{lemma}{Lemma}{Lemmas}

\newaliascnt{claim}{theorem}
\newtheorem{claim}[claim]{Claim}
\aliascntresetthe{claim}
\crefname{claim}{Claim}{Claims}

\newaliascnt{corollary}{theorem}
\newtheorem{corollary}[corollary]{Corollary}
\aliascntresetthe{corollary}
\crefname{corollary}{Corollary}{Corollaries}

\else

	\newaliascnt{claiml}{theorem}
	\newtheorem{claiml}[claiml]{Claim}
	\aliascntresetthe{claiml}
	
	\renewenvironment{claim}{\begin{claiml}}{\end{claiml}}

\fi

\newaliascnt{proto}{theorem}
\newtheorem{proto}[proto]{Protocol}
\aliascntresetthe{proto}
\crefname{proto}{Protocol}{Protocols}

\newaliascnt{algo}{theorem}
\newtheorem{algo}[algo]{Algorithm}
\aliascntresetthe{algo}
\crefname{algo}{Algorithm}{Algorithms}

\newaliascnt{construction}{theorem}

\aliascntresetthe{construction}
\crefname{construction}{Construction}{Constructions}

\newaliascnt{fact}{theorem}
\newtheorem{fact}[fact]{Fact}
\aliascntresetthe{fact}
\crefname{fact}{Fact}{Facts}

\iffull

\newaliascnt{proposition}{theorem}
\newtheorem{proposition}[proposition]{Proposition}
\aliascntresetthe{proposition}
\crefname{proposition}{Proposition}{Propositions}

\newaliascnt{conjecture}{theorem}

\aliascntresetthe{conjecture}
\crefname{conjecture}{Conjecture}{Conjectures}

\newaliascnt{definition}{theorem}
\newtheorem{definition}[definition]{Definition}
\aliascntresetthe{definition}
\crefname{definition}{Definition}{Definitions}

\fi

\newaliascnt{notation}{theorem}

\aliascntresetthe{notation}
\crefname{notation}{Notation}{Notation}

\newaliascnt{assertion}{theorem}

\aliascntresetthe{assertion}
\crefname{assertion}{Assertion}{Assertion}

\newaliascnt{assumption}{theorem}

\aliascntresetthe{assumption}
\crefname{assumption}{Assumption}{Assumption}
 
\iffull

\newaliascnt{remark}{theorem}
\newtheorem{remark}[remark]{Remark}
\aliascntresetthe{remark}
\crefname{remark}{Remark}{Remarks}

\newaliascnt{question}{theorem}
\newtheorem{question}[question]{Question}
\aliascntresetthe{question}
\crefname{question}{Question}{Questions}

\newaliascnt{example}{theorem}
\ifdefined\excludeexample
\else
   
\fi

\aliascntresetthe{example}
\crefname{exmaple}{Example}{Examples}

\fi

\crefname{equation}{Equation}{Equations}

 \newenvironment{protocol}{\begin{mybox} \vspace{-.1in}\begin{proto}}{ \vspace{-.1in} \end{proto}\end{mybox}}

\newenvironment{algorithm}{\begin{mybox} \vspace{-.1in}\begin{algo}}{ \vspace{-.1in} \end{algo}\end{mybox}}

\newenvironment{mybox}{\begin{center}\begin{tabular}{|p{0.97\linewidth}|c|}   \hline} {  \\ \hline \end{tabular} \end{center}}

\newaliascnt{expr}{theorem}
\newtheorem{expr}[expr]{Experiment}
\aliascntresetthe{expr}
\crefname{experiment}{experiment}{experiments}

\newcommand{\stepref}[1]{Step~\ref{#1}}

\def\FullBox{$\Box$}
\def\qed{\ifmmode\qquad\FullBox\else{\unskip\nobreak\hfil
\penalty50\hskip1em\null\nobreak\hfil\FullBox
\parfillskip=0pt\finalhyphendemerits=0\endgraf}\fi}

\def\qedsketch{\ifmmode\Box\else{\unskip\nobreak\hfil
\penalty50\hskip1em\null\nobreak\hfil$\Box$
\parfillskip=0pt\finalhyphendemerits=0\endgraf}\fi}

{\proof[#1]\list{}{\leftmargin=1cm\rightmargin=1cm}}
{\endproof%
  \vspace{10pt}
}

\newcommand{\eex}[2]{\Ex_{#1}\left[#2\right]}
\newcommand{\ex}[1]{\Ex\left[#1\right]}
\newcommand{\Ex}{{\mathrm E}}
\renewcommand{\Pr}{{\mathrm {Pr}}}
\newcommand{\pr}[1]{\Pr\left[#1\right]}
\newcommand{\ppr}[2]{\Pr_{#1}\left[#2\right]}

\newcommand{\Ac}{\MathAlgX{A}}
\newcommand{\Mc}{\MathAlgX{M}}
\newcommand{\Ec}{\MathAlgX{E}}

\newcommand{\Act}{{\tilde{\mathsf{A}}}}
\newcommand{\Bct}{{\tilde{\mathsf{B}}}}

\newcommand{\Bc}{\MathAlgX{B}}

\newcommand{\Dc}{\MathAlgX{D}}
\newcommand{\Gc}{\MathAlgX{G}}

\newcommand{\size}[1]{\left|#1\right|}

\newcommand{\prob}[1]{\mathsf{\textsc{#1}}}

\newcommand{\SD}{\prob{SD}}

\newcommand{\ppt}{{\sc ppt}\xspace}

\def\cD{{\cal D}}

\def\cG{{\cal G}}
\def\cH{{\cal H}}
\def\cI{{\cal I}}

\def\cR{{\cal R}}
\def\cS{{\cal S}}

\def\cU{{\cal U}}

\def\cX{{\cal X}}
\def\cY{{\cal Y}}

\def\bbN{\N}

\newcommand{\Tableofcontents}{
\thispagestyle{empty}
\pagenumbering{gobble}
\clearpage
\tableofcontents
\thispagestyle{empty}
\clearpage
\pagenumbering{arabic}
}

\newcommand{\OT}{{\sf OT}\xspace}

\DeclareMathOperator{\V}{V}

\newcommand{\Pc}{\MathAlgX{P}}

\renewcommand{\DP} {{\sf DP}\xspace}
\newcommand{\pk} {\kappa}
\newcommand{\OA} {O_\Ac}
\newcommand{\OB} {O_\Bc}

\newcommand{\tC} {\tilde{C}}
\newcommand{\ty} {\widetilde{y}}
\newcommand{\tY} {\widetilde{Y}}

\newcommand{\out}{{\rm out}}

\newcommand{\Lap}{{\rm Lap}}

\makeatletter
\let\xx@thm\@thm
\AtBeginDocument{\let\@thm\xx@thm}
\makeatother

\DeclareMathSymbol{\shortminus}{\mathbin}{AMSa}{"39}

\newcommand{\hY}{\widehat{Y}}

\newcommand{\Eve}{\Ec}

\newcommand{\CDP}{{\sf CDP}\xspace}
\newcommand{\KA}{{\sf KA}\xspace}

\newcommand{\sdist}[2]{\SD\paren{#1, #2}}

\newcommand{\hy}{\widehat{y}}

\newcommand{\keywords}[1]{\bigskip\par\noindent{\footnotesize\textbf{Keywords\/}: #1}}

\newcommand{\oA}{{o_\Ac}}
\newcommand{\vA}{{v_\Ac}}
\newcommand{\VA}{{V_\Ac}}

\newcommand{\oB}{{o_\Bc}}
\newcommand{\vB}{{v_\Bc}}
\newcommand{\VB}{{V_\Bc}}

\newcommand{\hVB}{{\widehat{V}_\Bc}}
\newcommand{\hVA}{{\widehat{V}_\Ac}}

\newcommand{\hoB}{{\widehat{o}_\Bc}}
\newcommand{\hoA}{{\widehat{o}_\Ac}}

\newcommand{\hOB}{{\widehat{O}_\Bc}}
\newcommand{\hOA}{{\widehat{O}_\Ac}}

\newcommand{\isize}{n}

\newcommand{\WEC}{{\sf WEC}\xspace}

\newcommand{\CompWEC}{{\sf CompWEC}\xspace}

\newcommand{\AWEC}{{\sf AWEC}\xspace}

\newcommand{\CompAWEC}{{\sf CompAWEC}\xspace}

\newcommand{\hAc}{\widehat{\Ac}}

\newcommand{\hBc}{\widehat{\Bc}}

\newcommand{\hPsi}{\widehat{\Psi}}

\newcommand{\pn}{n}

\newcommand{\hhC}{\widehat{C}}

\newcommand{\Dist}{\MathAlgX{Dist}}
\newcommand{\Pred}{\MathAlgX{Pred}}

\renewcommand{\nu}{n.u.\xspace}

\newcommand{\Nnote}[1]{\authnote{Noam}{#1}}
\newcommand{\Enote}[1]{\authnote{Eliad}{#1}}

\title{
Computationally Differentially Private Inner-Product Protocols Imply Oblivious Transfer
}

\iffull
\author{
Iftach Haitner\thanks{Stellar Development Foundation and The Blavatnik School of Computer Science, Tel Aviv University. {\tt iftachh@tauex.tau.ac.il}. 
Research supported by Israel Science Foundation grants 836/23.}
\and
Noam Mazor\thanks{Tel Aviv University. {\tt noammaz@gmail.com}. Research partially supported by NSF CNS-2149305, AFOSR
Award FA9550-23-1-0312 and AFOSR Award FA9550-23-1-0387 and ISF Award 2338/23.}
\and
Jad Silbak\thanks{Department of Computer Science, Northeastern University. {\tt jadsilbak@gmail.com}. Research supported by the Khoury College Distinguished Post-doctoral Fellowship.}
\and
Eliad Tsfadia\thanks{Department of Computer Science, Georgetown University. {\tt eliadtsfadia@gmail.com}. Research supported by a gift to Georgetown University.}
\and
Chao Yan\thanks{Department of Computer Science, Georgetown University. {\tt cy399@georgetown.edu}. Research supported by a gift to Georgetown University.}
}

\else

\author{
Iftach Haitner\inst{1,2}
\and
Noam Mazor\inst{2}
\and
Jad Silbak\inst{3}
\and
Eliad Tsfadia\inst{4}
\and
Chao Yan\inst{4}
}

\institute{Stellar Development Foundation \and
The Blavatnik School of Computer Science, Tel Aviv University\\
\email{iftachh@tauex.tau.ac.il, noammaz@gmail.com} \and
Department of Computer Science, Northeastern University\\
\email{jadsilbak@gmail.com} \and
Department of Computer Science, Georgetown University\\
\email{eliadtsfadia@gmail.com, cy399@georgetown.edu}
}

\fi

\begin{document}

\maketitle

\begin{abstract}

In \textit{distributed differential privacy}, multiple parties collaborate to analyze their combined data while each party protects the confidentiality of its data from the others. Interestingly, for certain fundamental two-party functions, such as the \textit{inner product} and \textit{Hamming distance}, the accuracy of distributed solutions significantly lags behind what can be achieved in the centralized model. For computational differential privacy, however, these limitations can be circumvented using \textit{oblivious transfer} (used to implement secure multi-party computation).
Yet, no results show that oblivious transfer is indeed necessary for accurately estimating a non-Boolean functionality. 
In particular, for the inner-product functionality, it was previously unknown whether oblivious transfer is necessary even for the best possible constant additive error.

In this work, we prove that any computationally differentially private protocol that estimates the inner product over $\oo^n \times \oo^n$ up to an additive error of $O(n^{1/6})$, can be used to construct oblivious transfer. In particular, our result implies that protocols with sub-polynomial accuracy are equivalent to oblivious transfer. In this accuracy regime, our result improves upon \citeauthor*{HaitnerMST22} [STOC '22] who showed that a key-agreement protocol is necessary.

\keywords{differential privacy; inner product; oblivious transfer}
\end{abstract}

\iffull
\Tableofcontents
\fi

\section{Introduction}
\textit{Differential privacy} \cite{DMNS06} is a mathematical privacy guarantee designed to facilitate statistical analyses of databases while ensuring strong guarantees against the leakage of individual-level information. 
Formally, 
\begin{definition}[Differential Privacy (DP)]\label{def:intro:DP}
	A randomized function (``mechanism'') $\Mc \colon \cX^n \rightarrow \cY$ is $(\eps,\delta)$-differentially private (in short, $(\eps,\delta)$-$\DP$) if for every $x,x' \in \cX^n$ that differ on only one entry, and any algorithm $\Dc\colon \cY \rightarrow \zo$ (``distinguisher''), it holds that
	\begin{align}\label{eq:intro:DP}
		\pr{\Dc(\Mc(x)) = 1} \leq e^{\eps}\cdot \pr{\Dc(\Mc(x')) = 1} + \delta.
	\end{align}
        When $\delta = 0$, we omit it from the notation.
\end{definition}

Most research in differential privacy regards the \textit{centralized model}: users share their data with a \emph{central server} that applies a differentially private mechanism and releases the results to the (potentially adversarial) world. In many scenarios, however, there is no such trusted server, and the data is distributed among several parties who wish to perform a joint computation while preserving the privacy of each party's data from the eyes of the others.
Constructing multiparty protocols that realize a functionality while preserving $\DP$ is typically more challenging than creating them in the centralized model;  a party's output should not only be a differentially private given the other parties' inputs and outputs (as in the centralized model), but also given the other parties' view in the protocol (\cite{DN04,BNO08}).

A  fundamental example that illustrates the hardness of distributed $\DP$ is the inner-product functionality of two binary datasets:
parties $\Ac$ and $\Bc$ hold datasets $x \in \oo^n$ and $y \in \oo^n$ (respectively), and wish to estimate the inner product over the integers $\ip{x,y} = \sum_{i=1}^n x_i \cdot y_i$ (the correlation between the datasets) using a $\DP$ protocol.
There exists a protocol (specifically, the so-called \textit{randomized response})  that is $\eps$-$\DP$ and estimates the inner-product up to additive error $O_{\eps}(\sqrt{n})$ (\cite{McGregorMPRTV10,MPRV09}), and this essentially the best error rate we can achieve while maintaining $\eps$-$\DP$: \cite{McGregorMPRTV10,HaitnerMST22} showed that any $\eps$-$\DP$ protocol for the inner-product, even when allowing any $\delta \in o(1/n)$, must have an additive error of $\Omega_{\eps}(\sqrt{n})$. This is in contrast to the centralized model, where the simple $\eps$-$\DP$ Laplace mechanism $\Mc(x,y) = \ip{x,y} + \Lap(2/\eps)$ has only additive error $O_{\eps}(1)$. 

To achieve better accuracy, \cite{BNO08,MPRV09} considered protocols that only guarantee \emph{computational differential privacy ($\CDP$)}: a computational analog of  $\DP$ where \cref{def:intro:DP} is relaxed so that  \cref{eq:intro:DP} only holds for computationally efficient (\ppt) algorithms $\Dc$, and not necessarily for any algorithm $\Dc$. This relaxation is very powerful. For example, if we assume the existence of the cryptographic primitive \emph{oblivious transfer} (in short, $\OT$, \cite{Rabin81}), then in any distributed setting (even for more than two parties), we can achieve the accuracy of a centralize $\DP$ mechanism by emulating it using \emph{secure multiparty computation} (\cite{BNO08,DKMMN06}). In particular, using $\OT$, two parties can securely emulate the (centralized) Laplace mechanism to achieve an $\eps$-$\CDP$ protocol with additive error $O_{\eps}(1)$. This leads to the following fundamental question:

\begin{question}
	Can we construct a two-party $\CDP$ protocol that estimates the inner product with a small additive error using a cryptographic primitive that is \emph{weaker} than $\OT$? 
\end{question}

\citet*{HaitnerMST22}  made partial progress on this matter, showing that any $(\eps,\frac1{n^2})$-$\CDP$ protocol with non-trivial accuracy guarantee of $o_{\eps}(\sqrt{n})$ must use a cryptographic primitive that is at least as strong as \emph{key agreement} (in short, $\KA$).  However, although $\KA$ and $\OT$ are both public-key cryptography primitives, $\KA$ is believed to be a weaker primitive (\cite{GertnerKMRV00}) and does not suffice for secure MPC. Even for $\CDP$ protocols that estimate the inner-product with very high accuracy of $O_{\eps}(1)$, it was previously unknown whether $\OT$ is indeed necessary.

\subsection{Our Result}
We prove the first equivalence between $\OT$ and $\CDP$ two-party protocols that estimate the inner-product with small additive error, and in particular, additive error up to $O_{\eps}(n^{1/6})$.

\begin{theorem}[Main result, informal]\label{thm:intro:main}
    An $\paren{\eps,\frac1{3n}}$-$\CDP$ two-party protocol that estimates the inner product over $\oo^n \times \oo^n$ up to an additive error $n^{1/6}/(c\cdot e^{c \eps})$ with probability $0.999$ (for some universal constant $c > 0$), can be used to construct an $\OT$ protocol.\footnote{\cref{thm:intro:main} is proven via a uniform and black-box reduction. Thus, although the security of $\CDP$ and $\OT$ is commonly defined with respect to non-uniform adversaries, \cref{thm:intro:main}  also implies a reduction from $\OT$ that is secure against uniform adversaries to such mildly accurate inner-product protocols that are $\CDP$ with respect to uniform distinguishers.}
\end{theorem} 

In particular, \cref{thm:intro:main} implies that without assuming the existence of $\OT$, it is impossible to construct a $\CDP$ protocol for the inner-product with a constant, or even sub-polynomial, additive error. It follows that a protocol with a small polynomial additive error is equivalent to a protocol with a constant error. That is because using such a mildly accurate protocol, we can construct by \cref{thm:intro:main} an oblivious transfer, and then use it to perform a secure two-party computation of the Laplace mechanism.
\begin{corollary}
There exists a universal constant $c>0$ such that the following holds: assume that an $\paren{\eps,\frac1{3n}}$-$\CDP$ two-party protocol that estimates the inner product over $\oo^n \times \oo^n$ up to an additive error $n^{1/6}/(c\cdot e^{c \eps})$ with probability $0.999$ exists, then there exists an $\paren{\eps,\negl(n)}$-$\CDP$ protocol that estimates the inner product up to an additive error $20/\eps$ with probability $0.999$.
\end{corollary}
\cref{thm:intro:main} is the first equivalence result between $\OT$ to any $\CDP$ protocol that estimates a natural, non-Boolean functionality.
Prior works only provide such equivalences for $\CDP$ protocols that estimate the (Boolean) XOR functionality (\cite{GKMPS16,HMSS19}), a less interesting functionality in the context of $\DP$, as even in the centralized model, any $\DP$ algorithm for estimating XOR must incur an error close to one-half.
In contrast, the inner product has a significantly larger gap between the accuracy achievable in two-party $\DP$ and in $\CDP$.
A proof overview of \cref{thm:intro:main} is given in \cref{sec:overview}.

\paragraph{Comparison with \cite{HaitnerMST22}.}
We remark that our result is not strictly stronger than \cite{HaitnerMST22}. First, \cite{HaitnerMST22} covered the whole non-trivial accuracy regime up to $O_{\eps}(\sqrt{n})$ while our result is only limited to $O_{\eps}(n^{1/6})$. Second, \cite{HaitnerMST22}'s result holds even for low confidence regimes, while \cref{thm:intro:main} requires high accuracy confidence of $0.999$. But in the natural high-accuracy high-confidence regimes, our result has several advantages beyond the necessity of a stronger primitive. First, \cite{HaitnerMST22}'s reduction only works for $\CDP$ two-party protocols that provide the inner-product estimation within the transcript, while our reduction in \cref{thm:intro:main} just assumes that (at least) one of the parties sees the estimation. Second, and perhaps surprisingly, the proof of \cref{thm:intro:main} is conceptually simpler, even though it provides reduction from a stronger primitive. We provide more details in \cref{sec:overview} (see \cref{remark:OTvsKA}). 

\remove{
\Nnote{the paragraph below is hard to follow, need to think what to do with it - maybe move it to be a high-level remark after the protocol is the proof overview? I'm not sure we want to say the words SV, condenser, extractor, ...}
Very roughly, in the secrecy analysis of \cite{HaitnerMST22}'s $\KA$ protocol, they had to prove (in a constructive way) that the inner product of a strong Santha Vazirani (SV) source \cite{McGregorMPRTV10} with a uniformly random seed is a good \emph{condenser}, even when the seed and the source are \emph{dependent}. When the seed and the source are \emph{independent}, \cite{McGregorMPRTV10} already proved that the inner product is not only a good condenser for strong SV sources, but also a good \emph{extractor} for (standard) SV sources \cite{SV87}, which apart of their non-constructive technique and the independency assumption, is a stronger result. However, \cite{HaitnerMST22}'s dependency between the source and the seed resulted in a significant technical challenge, which made the proof conceptually harder. In contrast, in our $\OT$ security analysis, since we only provide secrecy guarantees from the point of views of the parties (and not from the point of view of an external observer that only sees the transcript), we bypass this dependency issue. See more details in \cref{sec:overview}.
}


\subsection{Open Questions}

In this paper, we show that a $\paren{\eps,\frac1{3n}}$-$\CDP$ protocol that estimates the inner product over $\oo^n \times \oo^n$ up to an additive error $O_{\eps}(n^{1/6})$ is equivalent to oblivious transfer ($\OT$). For the inner-product functionality, it is still left open to close the gap up to any non-trivial accuracy of $O_{\eps}(\sqrt{n})$. \cite{HaitnerMST22} closed this gap w.r.t.\ a weaker notion of $\CDP$ \emph{against external observer}\footnote{A protocol is $\CDP$ against external observer, if the transcript of the execution is a $\CDP$ function of the input datasets, but not necessarily the views of the parties as in the standard notion.} by showing that key agreement ($\KA$) is necessary and sufficient for any non-trivial accuracy guarantee, which in particular implies that $\KA$ is also necessary for the standard notion of $\CDP$. But since $\KA$ is strictly weaker than $\OT$ (under block-box reductions), we still do not know what is the right answer for additive errors in the range between $n^{1/6}$ and $\sqrt{n}$.

Beyond characterizing the inner-product functionality, the main challenge is to extend this understanding to other $\CDP$ distributed computations. For some functionalities, \eg Hamming distance, we have a simple reduction to the inner-product functionality. But finding a more general characterization that captures more (or even all) functionalities, remains open.

\subsection{Paper Organization}

\iffull
In \cref{sec:overview}, we give a high-level proof of \cref{thm:intro:main}, and in \cref{sec:related-works}, we discuss additional related works. Notations, definitions, and general statements used throughout the paper are given in \cref{sec:preliminaries}.
Our main formal theorems are stated in \cref{sec:main_theorems}, and the oblivious transfer protocol together with its
security proof are given in
\cref{sec:DPIP_to_WAEC}. 
In \cref{sec:related-works} we discuss additional related works.
The proof of our main result relies on a technical tool that is proven in \cref{sec:AWEC-to-WEC}. Other missing proofs appear in \cref{sec:missing-proofs}.

\fi

\section{Proof Overview}\label{sec:overview}
In this section, we give an overview of our proof techniques. Our starting point is the key-agreement protocol of \cite{HaitnerMST22}. As mentioned above, \cite{HaitnerMST22} showed that a $\CDP$ protocol $\Pi$ for computing the inner product functionality implies the existence of key-agreement ($\KA$). To prove that, they used $\Pi$  to construct a (weak) $\KA$ protocol, in which the parties $\Ac$ and $\Bc$ interact as follows:
\begin{enumerate}
    \item In the first step, $\Ac$ and $\Bc$ choose random inputs $x\in \set{-1,1}^n$ and $y\in \set{-1,1}^n$, respectively.
    \item The parties interact using $\Pi$ to get approximation $z$ of $\langle x,y \rangle$.
    \item Next, $\Ac$ chooses a random string $r \gets \zn$, and sends $r, x_r=\set{x_i \colon r_i=1}$ to $\Bc$. $\Bc$ replies with $y_{- r}=\set{y_i \colon r_i=0}$.
    \item Finally, $\Ac$ computes and outputs $\out_\Ac=\langle x_{-r},y_{-r} \rangle$, and $\Bc$ outputs $\out_\Bc=z-\langle x_{r},y_{r} \rangle$.
\end{enumerate}
The first observation made by \cite{HaitnerMST22} is that, by linearity of the inner product, $\langle x_{-r},y_{-r} \rangle+\langle x_{r},y_{r} \rangle=\langle x,y \rangle$, and thus $\size{\out_A-\out_B}=\size{z-\langle x,y \rangle}$, which is small assuming $\Pi$ is accurate. Moreover, it was shown that this protocol can be amplified into a full-fledged $\KA$ protocol if any efficient adversary $\Eve$ cannot approximate $\out_A$ within a small additive distance, given the transcript of the weak $\KA$ protocol. The main technical part in \cite{HaitnerMST22} is to show that this is indeed the case in the above protocol. 

In this work, we are interested in constructing a stronger primitive, namely oblivious transfer ($\OT$). In contrast to a $\KA$ protocol in which security must hold against external observers that only see the transcript, in $\OT$, the security needs to hold against an adversary that has access to the full view of one of the parties. 
 Our first observation is that the original $\KA$ protocol of \cite{HaitnerMST22}, as is, 
cannot be used to construct any form of even ``weak" oblivious transfer.
In more details, it is possible to show that for some carefully chosen $\CDP$ and accurate protocol  $\Pi$, the joint view of the parties in the above protocol can be \emph{simulated} using a trivial protocol, without using $\Pi$ at all.\iffull
\footnote{Here a protocol $\Pi'$ simulates $\Pi$ if $\Pi'$ generates the joint view of the parties $\Ac$, $\Bc$ in $\Pi$, without revealing any other information to the parties. We say that the simulation $\Pi'$ is trivial if it does not rely on any nontrivial protocol or cryptographic assumptions. Importantly, the simulation may reveal new information (that is already known for the internal parties $\Ac$ and $\Bc$) to an external observer, and thus it doesn't imply that the protocol is not a $\KA$ protocol. See \cref{appendix:HaitnerMST22} for a concrete example.}
\fi
Since there is no $\OT$ without computational assumptions, this leads to the conclusion that the $\KA$ protocol of \cite{HaitnerMST22} could not be used as a subroutine (in a black-box manner) to construct $\OT$. Looking ahead, our solution is to inject carefully placed \emph{noise} to the interaction, in a manner that the parties cannot simulate without revealing private information. As we will see shortly,  this enables us to bypass the limitations of the original construction of \cite{HaitnerMST22} and to achieve $\OT$.

\paragraph{Weak Erasure Channel.} Before describing our construction, let us explain the properties that we need the construction to fulfill. Similarly to \cite{HaitnerMST22}, we will start with constructing a weak $\OT$ protocol, and then we will show how to amplify it into a full-fledged $\OT$ protocol. By the work of \cite{Wullschleger09}, to achieve OT, it is sufficient to construct a \emph{weak erasure channel} (\WEC).\footnote{A weak erasure channel is a weak version of Rabin's $\OT$ \cite{Rabin81}, which is known to be equivalent to ${2 \choose 1}$-$\OT$.} In this work, we construct a protocol with a weaker security guarantee, which we call \emph{approximate weak erasure channel} (\AWEC), and show that such a protocol also suffices for constructing $\OT$. Informally, an $(\ell,\alpha,p,q)$-\AWEC is a no-input protocol between   $\Ac$ and $\Bc$, party $\Ac$ outputs a number $O_\Ac\in [-n,n]$, and party $\Bc$ either outputs a number $O_\Bc\in [-n,n]$, or an erasure symbol $O_\Bc=\bot$. We additionally require that:

\begin{itemize}
		
		\item Erasure happens with probability $1/2$. Namely,   $\pr{O_\Bc = \perp} = 1/2$.
		
		\item When there is no erasure, $\Ac$ and $\Bc$ approximately agree. More formally, $$\pr{\size{O_A - O_B} > \ell \mid O_B \ne \bot}\le \alpha.$$
		
		\item When an erasure occurs, $\Bc$ cannot predict the value of $O_A$:  let $V_B$ be the view of $\Bc$ in the protocol, then for any efficient $\widehat{\Bc}$ 
        \begin{align*}
				\pr{\size{\widehat{\Bc}(V_B) - O_A} \leq 1000 \ell \mid O_B=\bot} \leq q.
			\end{align*}
			(i.e., if $O_B=\bot$, Bob can't estimate the value of $O_A$ with error $\leq 1000 \ell$).
	
        \item Lastly, we require that $\Ac$ cannot tell when an erasure occurs: let $V_A$ be the view of $\Ac$ in the protocol, then for any efficient   
		$\widehat{\Ac}$ 
			\begin{align*}
				\size{\pr{\widehat{\Ac}(V_A) = 1 \mid O_B \neq \perp} - \pr{\widehat{\Ac}(V_A) = 1 \mid O_B = \perp}} \le p.
			\end{align*}
	\end{itemize}
  See \cref{def:AWEC} in \cref{sec:AWEC} for the formal definition of \AWEC.  
  We show that when $p+q+\alpha \ll 1$, an \AWEC protocol can be amplified into an $\OT$. We do this by showing a reduction from \AWEC to \WEC and then applying the amplification result of \cite{Wullschleger09} to get a full-fledged $\OT$.

\paragraph{The Construction.} We next describe how to use a $\CDP$ protocol $\Pi$ computing the inner-product to construct an \AWEC. Our accuracy requirement from $\Pi$ is that on random inputs, $\Pi$ output is within distance $\ell \ll n^{1/6}$ from the inner-product with probability at least $0.999$. That is,
\[ 
\ppr{x,y\gets\oo^n,
z\gets \Pi(x,y)}{ \size{z-\langle x,y\rangle}\le \ell} \ge 0.999.
\]
The construction is similar to the construction of \cite{HaitnerMST22}, with the key difference that with probability $1/2$, $\Bc$ sends a noisy version of its vector to $\Ac$. Informally, this random noise ``erases" the information $\Bc$ has on $\Ac$'s output.

\begin{protocol}[\AWEC]\label{protocol:intro:DPIP-to-AWEC}
	\item Operation:
	\begin{enumerate}
		
		\item $\Ac$ samples $x\gets \oo^n$, and $\Bc$ samples $y\gets \oo^n$.
        \item $\Ac$ and $\Bc$ interact according to $\Pi$ using inputs $x$ and $y$ \resp to get output $z\approx \langle x,y \rangle \pm \ell$. 
		\item  $\Ac$ samples  $r \la \zo^n$ and sends $(r,x_{r} = \set{x_i \colon r_i =1})$ to $\Bc$.
		
		\item $\Bc$ samples a random bit $b \la \zo$ and acts as follows:
		
		\begin{enumerate}
			
			\item If $b=0$ (``non-erasure''), it sends $y_{-r}= \set{y_i \colon r_i =0}$ to $\Ac$ and outputs $o_B = z - \ip{x_r, y_r}$.
			
			\item Otherwise ($b=1$, ``erasure''), it performs the following steps:
			
			\begin{itemize}
				\item Samples $k$ uniformly random indices $i_1,\ldots,i_k \la [n]$, where $n^{1/3} \gg k \gg  e^{\eps} \cdot  \ell^2$.
                
				\item Compute $\ty = (\ty_1,\ldots,\ty_n)$ where $\ty_i \la \begin{cases} \cU_{\oo} & i \in \set{i_1,\ldots,i_k} \\ y_i & \text{otherwise} \end{cases}$,
				\item Send $\ty_{-r}= \set{\ty_i \colon r_i =0}$ to $\Ac$, and
				\item Output $o_B = \perp$.
			\end{itemize}

		\end{enumerate}
		
		\item  Denote by $\hy_{-r}$ the message $\Ac$ received from $\Bc$. Then $\Ac$ outputs $o_A =  \ip{x_{-r}, \hy_{-r}}$.

	\end{enumerate}
\end{protocol}
That is, with probability $1/2$, $\Bc$ decides to make an erasure by changing the value of $k$ random bits of $y$ before sending $y_{-r}$ to $\Ac$. Let $\mu=\widetilde{y}-y$ be the noise vector. Then importantly, the crux of the above protocol is that the noise $\langle \mu_{-r}, x_{-r} \rangle = \langle \widetilde{y}_{-r}, x_{-r} \rangle-   \langle y_{-r}, x_{-r} \rangle$ cannot be simulated by $\Ac$ or $\Bc$ without revealing more information on the input of the other party. We next give an overview of the analysis of the above construction. 

\paragraph{Analysis.} 
We define the following random variables \wrt a random execution of \cref{protocol:intro:DPIP-to-AWEC}. Let $X,Y$ be inputs of $\Ac$ and $\Bc$ to $\Pi$, respectively, and let $V_A$ and $V_B$ be their view in the interaction according to $\Pi$.  Let $R, I_1,\dots, I_k, \widetilde{Y}$ and $\widehat{Y}$ be the random variables taking the values of $r,i_1,\dots,i_k,\widetilde{y}$ and $\widehat{y}$, respectively. 
Finally, let $\widehat{V}_A =(V_A, R, \widehat{Y}_{-R})$ and $\widehat{V}_B= (V_B, R, X_R, I_1,\dots, I_k, \widetilde{Y},\widehat{Y})$ be the parties' views in \cref{protocol:intro:DPIP-to-AWEC}, and $O_A,O_B$ their outputs, respectively. 

It is evident by the definition of the protocol that erasure indeed occurs with probability $1/2$. Moreover, similarly to the correctness of the key-agreement protocol, and by the accuracy of $\Pi$, when there is no erasure, $O_A$ and $O_B$ are at most $\ell$ apart with probability at least $0.999$. Thus, it remains to prove that $\Bc$ cannot approximate $O_A$ when erasure occurs, and that $\Ac$ cannot learn whether it happened. To simplify the analysis in this proof overview, in the following, we assume that $\Pi$ is a $(\eps,0)$-\CDP protocol,\footnote{While $(\eps,0)$-\CDP (that is, setting $\delta=0$) doesn't make sense in the computational setting, for the sake of simplicity, we ignore this subtly in this informal overview.}  for a small constant $\eps>0$.

Note that while $\Bc$ learns the same information on $\Ac$’s input in the protocol in both the erasure and non-erasure cases, $\Ac$’s output $O_A$ is different in each case; when there is no erasure, $O_A$ is predictable (up to an additive error $\ell$) given $\Bc$’s view $V_B$. When there is an erasure, we claim that the privacy guarantee of $\Pi$ implies that $O_A$, given $V_B$, is unpredictable up to an additive error $k^{1/2} \gg \ell$ (more details in \cref{sec:techniques:erasure-do-erase}). The “magic” of the protocol is in the fact that, for $\ell \ll n^{1/6}$, $\Bc$ can control $\Ac$’s output, and manipulate it to be unpredictable (from $\Bc$’s point of view), without $\Ac$ noticing the difference. Proving the latter property (hidden erasure) is the main technical challenge of this work, and we provide an overview of the proof in \cref{sec:techniques:erasures-are-hidden}. By combining the two requirements on $k$, we obtain that our protocol is indeed $\AWEC$.

\begin{remark}[Oblivious transfer v.s. key agreement]\label{remark:OTvsKA}
Before discussing the security proof of our protocol, we would like to highlight an essential difference between the analysis of our protocol and \cite{HaitnerMST22}'s $\KA$ protocol. While the primitive we construct here, $\OT$, is a stronger primitive than $\KA$, certain aspects of the proof become more manageable in our analysis. This is because in $\OT$, we need to consider an internal attacker, while in $\KA$ the security needs to hold against an external observer. In more detail, \cite{HaitnerMST22} requires showing that security holds against an adversary who sees $R, X_{R}$, and $Y_{-R}$. This form of leakage creates a dependency between the seed $R$ and both the leakage on  $X$, $X_R$, and the leakage of $Y$, $Y_{-R}$, which makes the analysis challenging. In contrast, in our setting, the attacker sees one of the vectors, $X$ or $Y$, entirely. For example, $\Ac$ gets to see $X, R$ and $Y_{-R}$, and therefore there is no dependency between the leakage on $X$ and $R$. This fact helps us to bypass this dependency issue and makes this part of our analysis more manageable.
\end{remark}

\subsection{Erasures do Erase}\label{sec:techniques:erasure-do-erase}  

To see why ``erasures'' in \cref{protocol:intro:DPIP-to-AWEC} do erase the value of $O_A$ from $\Bc$'s point of view, let’s first consider a different version of \cref{protocol:intro:DPIP-to-AWEC}, where in the erasure case, party $\Bc$, instead of changing $k$ random bits in $Y_{-R}$, changes \emph{all} the bits to random ones. (Namely, in the erasure case, $\hY_{-R}$ is a uniformly random string, independent of $X$). Now, the question is, can $\Bc$ now predict $O_A = \ip{X_{-R}, \hY_{-R}}$? From $\Bc$’s point of view, the privacy guarantee of the protocol $\Pi$ implies that $X_{-R}$ must be “unpredictable”, and in \cite{McGregorMPRTV10,HaitnerMST22}, it is called “strong Santha Vazirani (SV)” source. \cite{McGregorMPRTV10} showed that the inner product of an $n$-size strong SV source with a random string, even if the string is known, is unpredictable up to an additive error of $\approx \sqrt{n}$, and \cite{HaitnerMST22} extended this result to the computational case. Therefore, if $\Bc$ decides to erase in this version, it cannot estimate $O_A = \ip{X_{-R}, \hY_{-R}}$ with an error smaller than $\approx \sqrt{n}$, which is much larger than $\ell$, its estimation of $\ip{X_{-R}, \hY_{-R}}$ in the non-erasure case. The problem with this version is that $\Ac$ can distinguish between the erasure and non-erasure cases. Hence, in the actual protocol, the number of bits that $\Bc$ changes in the erasure case is set to some smaller parameter $k$. The point is that the same argument we applied to the previous version still holds; if $\Bc$ changes $k$ random values in $Y_{-R}$ (call this $k$-size set of indices by $\cH$), then $O_A$ would be a linear function of $\ip{X_{\cH},\hY_{\cH}}$, and since $\Bc$ cannot estimated $\ip{X_{\cH},\hY_{\cH}}$ with error smaller than $\approx \sqrt{k}$ (from the same arguments above), it cannot do it for $O_A$ as well.

\iffull

In more detail, assume towards a contradiction that there exists an efficient  $\widehat{\Bc}$, that given $\widehat{V}_B$ and conditioned on $O_B=\bot$ (equivalently, on $b=1$)  can approximate $O_A$ within a small distance. Namely, assume that
$$
\pr{\size{\widehat{\Bc}(\widehat{V}_B) - O_A} \leq 1000 \ell \mid O_B=\bot} >q 
$$
for some small constant $q$. Our goal is to construct, using $\widehat{\Bc}$, an algorithm $\widetilde{\Bc}$ that breaks the \CDP guarantee of $\Pi$.

Building on \cite{McGregorMPRTV10,HaitnerMST22}, to break the \CDP property of the protocol $\Pi$, it is enough to construct an algorithm that approximates the inner product of (a subset of) $X$ with a random vector $S$. That is, we want to construct an algorithm $\widetilde{\Bc}$ that for a random subset $\cH\subseteq [n]$ of size $k/4$ and for a random vector $S\gets \oo^{k/4}$,  approximates the inner product $\langle X_{\cH}, S \rangle$ within small additive distance. More formally, we want an algorithm $\widetilde{\Bc}$ such that 
$$\pr{\size{\widetilde{\Bc}(V_B,\cH,S, X_{-\cH})-\langle X_{\cH}, S \rangle}\le \sqrt{\size{\cH}/e^\eps}}> \alpha.$$
The work of \cite{McGregorMPRTV10} implies that in the information-theoretic settings, such an algorithm contradicts the privacy guarantee of $\Pi$.\footnote{\cite{McGregorMPRTV10} showed that when $\cH=[n]$, approximating the inner product contradicts DP. Extending this result for any subset $\cH$ is not hard.} \cite{HaitnerMST22} extended this result to the computational settings, by giving a constructive proof. Thus, such an algorithm $\widetilde{\Bc}$ is enough to break the assumed \CDP property of $\Pi$.

To construct $\widetilde{\Bc}$, we first need to make an observation regarding  $\widehat{\Bc}$. 
Let $\cH=\set{I_1,\dots,I_k} \cap \set{i\colon R_i=0}$ be the set of indices that $\Bc$ sends to $\Ac$ and replace the bit with a random value, and assume for simplicity that $\size{\cH}=k/4$.\footnote{By a simple concentration bound, the size of $\cH$ is at least $k/4$ with all but negligible probability. When $\cH$ is larger, we consider the first $k/4$ indices in $\cH$.} Let $S=\widetilde{Y}_\cH$, and note that the distribution of $S$ is uniformly random, even given $\cH,X$ and $V_B$. Our main observation is that if $\widehat{\Bc}$  approximates $O_A=\langle X_{-R},\widetilde{Y}_{-R} \rangle$, then it must also compute a good approximation of $\langle X_{\cH}, S \rangle =\langle X_{\cH},  \widetilde{Y}_{\cH}\rangle$.

To see why, let $\bar{\cH}=\set{i\colon R_i=0}\setminus \cH$, and let   $T=\widehat{\Bc} (\widehat{V}_B)-\langle X_{\bar{\cH}},\widetilde{Y}_{\bar{\cH}} \rangle$. We claim that $T$ is a good approximation of $\langle X_{\cH}, S \rangle$. Indeed,
\begin{align*}
T-\langle X_{\cH}, S \rangle &= T-(\langle X_{-R},  \widetilde{Y}_{-R}\rangle - \langle X_{\bar{\cH}},  \widetilde{Y}_{\bar{\cH}}\rangle)\\
&= (T + \langle X_{\bar{\cH}},  \widetilde{Y}_{\bar{\cH}}\rangle)-\langle X_{-R},  \widetilde{Y}_{-R}\rangle \\
& =\widehat{\Bc}(V_B) -\langle X_{-R},  \widetilde{Y}_{-R}\rangle. 
\end{align*}
Therefore, $\size{T-\langle X_{\cH}, S \rangle}\le \sqrt{\size{\cH}/e^\eps}$ if $\size{\widehat{\Bc}(V_B) -\langle X_{-R},  \widetilde{Y}_{-R}\rangle}\le 1000\ell$ (note that by our choice of $k$, $\ell \ll \sqrt{\size{\cH}/e^\eps}=\sqrt{k/4e^\eps}$).

Thus, we only need to show an algorithm $\widetilde{\Bc}$ that, given $V_B,\cH,S,X_{-\cH}$ as input, evaluates the value of $T$. The main issue left is that while $\widetilde{\Bc}$ only gets $V_B$, and to evaluate $T$ we need to execute $\widehat{B}$ that needs to get an entire view $\widehat{V}_B$ of $\Bc$ in \cref{protocol:intro:DPIP-to-AWEC}. However, this can be easily simulated. Indeed, $\widetilde{\Bc}$ can simply sample $R,I_1,\dots,I_k$ and $\widetilde{Y}$ condition on the values of $\cH, S$ (and $O_B=\bot$), to get the exact distribution of $\widehat{V}_B$.

\fi

\subsection{Erasures are Hidden}\label{sec:techniques:erasures-are-hidden}

This part of the proof is more involved and is our main new technical part. The goal is to show that $\widehat{\Ac}$ cannot tell when an erasure happens; equivalently, that $\widehat{\Ac}$ cannot distinguish $Y_{-R}$ from the noisy version $\widetilde{Y}_{-R}$ with too good probability. For this end, we show that an algorithm $\widehat{\Ac}$ that (given $V_A$) can distinguish $Y_{-R}$ from the noisy version $\widetilde{Y}_{-R}$ with a small but constant probability, can be used to break $\CDP$. 

\remove{
and is achieved by reduction: We show that if $\Ac$ can distinguish, then it can be used (in a black-box manner) to construct a prediction algorithm $\Gc$ that, given $n-1$ bits, predicts the missing bit too accurately, thereby violating the protocol's privacy guarantee. By a hybrid argument, if Alice can distinguish if $k$ random bits were changed (with constant probability), then she can distinguish with probability $\approx 1/k$ if a single random bit has changed. So the predictor $\Gc$ uses $\Ac$ distinguishing to predict a random bit correctly with probability $\approx 1/k$. The key step is to perform this prediction in a way that the probability that $\Gc$ predicts the missing bit incorrectly is $\ll 1/k$. We do that by carefully leveraging the fact that $\Ac$ distinguishing succeeds while only seeing a subset of the input (and not $n-1$ bits like $\Gc$ sees). \remove{In the proof of Lemma 2.3 (or 5.3), we use Alice to create the three values \mu^{-1}, \mu^1, \mu^* such that if b \in {0,1} is the right value of the missing bit, then |\mu^b - \mu^*| > 1/k and |\mu^{-b} - \mu^*| < k^2/n, so as long as 1/k >> k^2/n, it is easy to predict the right b. The latter holds when k << n^{1/3} }
}

The first step is to notice that, by a rather standard hybrid argument, it is enough to show that $\widehat{\Ac}$ cannot distinguish with advantage $\Theta(1/k)$ between $Y$ and a vector $Y'$ which is identical in all but one random coordinate. That is, let $I\gets [n]$ be a random index, and let  $Y^{(I)}=(Y_1,\dots,Y_{I-1},-Y_I,Y_{I+1},...,Y_n)$ be  the vector that is identical to $Y$ in all entries except the $I$-th one. Then it is enough to show that if
\[ \pr{\widehat{\Ac}(V_A,R,Y_{-R})=1}-\pr{\widehat{\Ac}(V_A,R,Y^{(I)}_{-R})=1} \ge 1/k,\]
  then $\widehat{\Ac}$ can be used to break the $\eps$-$\CDP$ of $\Pi$.

We highlight that the key difficulty in proving the above lies in the fact that the distinguishing advantages are very small. Indeed, an algorithm $\widehat{A}$ that distinguishes $Y_{-R}$ from $Y^{(I)}_{-R}$ with a large enough \emph{constant} advantage (proportional to $\eps$) already contradicts the privacy requirement of $\Pi$. However, in our setup, the distinguishing advantage is much smaller than $\eps$, which makes the task of contradicting $\eps$-\CDP challenging. We next emphasize two crucial points that make $\widehat{\Ac}$ useful (and our design for our \AWEC protocol, \cref{protocol:intro:DPIP-to-AWEC}, was guided by these two observations).

\begin{description}
    \item[The importance of a random index]  We note that it is crucial that the index $I$ is chosen at random and that it is unknown to $\widehat{\Ac}$. To see that, let $\Mc$ be a \DP mechanism and let $Y\gets \oo^n$ be an input. We claim that it can be the case that an adversary can distinguish between $Y$ and $Y^{(i)}$ given $\Mc(Y)$ with advantage $\approx \eps$, for any fixed index $i$. Indeed, it is possible that $\Mc(Y)$ outputs a noisy estimation of $Y_i$, which is correct with probability $\approx 1/2 + \eps$. Such an estimation will not contradict differential privacy, but will be useful to distinguish between $Y$ and $Y^{(i)}$. We resolve this issue by adding the noise in a random position. To prove that the existence of such noise (in a random position) is indeed hidden, we developed and proved \cref{lemma:property1:prediction:overview} which is our main technical lemma.
   
    \item[The importance of a random subset] Next, we claim that it is crucial that $\widehat{\Ac}$ only gets $Y_{-R}$, a random subset of $Y$. Indeed, as in the previous example, the mechanism $\Mc(Y)$ can output a noisy estimation of the parity of the number of ones in $Y$. As in the previous example, such an estimation does not contradict \DP, but is enough to distinguish $Y$ and $Y^{(I)}$ by considering the number of ones in the input compared to the estimation. We deal with such a \emph{global} information by revealing to $\widehat{\Ac}$ only a subset of $Y$. This makes the global information not useful.
\end{description}

Thus, in the proof, we need to leverage the fact that $\widehat{\Ac}$ distinguishes $Y_{-R}$ from $Y^{(I)}_{-R}$ for a random index $I$ and a random subset $R$. We use such an algorithm $\widehat{\Ac}$ to construct an algorithm $\Gc$ that given $Y_{-I} = (Y_1,\ldots,Y_{I-1}, Y_{I+1},\ldots,Y_n)$, predicts the value of $Y_I$. In more detail, we prove the following lemma, which is the technical core of our proof.

\begin{lemma}[\cref{lemma:property1:prediction}, informal]\label{lemma:property1:prediction:overview}
	Let $(Z,W)$ be jointly distributed random variables, let $R \la \zo^n$ and $I \la [n]$. 
	Let  $\Ac$ be an efficient algorithm that satisfies
	\begin{align*}
		\size{\ex{\Ac(R,Z_{R},W) - \Ac(R,Z^{(I)}_{R},W)}} \geq 1/k,
	\end{align*}
	for $Z^{(I)} = (Z_1,\ldots, Z_{I-1}, -Z_I, Z_{I+1},\ldots, Z_n)$. Then, there exists an efficient algorithm $\Gc$ that outputs a value in $\set{1,-1,\bot}$, such that the following holds.
		$$\pr{\Gc(I,Z_{-I}, W) = -Z_I} \leq \frac{k^2}{n},\ and,\ \pr{\Gc(I,Z_{-I}, W) = Z_I} \geq 1/k - k^2/n.$$
\end{lemma}

Back to our setting, by taking $\Ac(R, Z_R, W) = \widehat{\Ac}((1,\ldots,1)- R, Z_R, W)$, $Z=Y$ and $W=V_A$ in \cref{lemma:property1:prediction:overview}, we get an algorithm $\Gc$ such that 
$$\pr{\Gc(I,Y_{-I}, V_A) = -Y_I} \leq \frac{k^2}{n},\ and,\
		\pr{\Gc(I,Y_{-I}, V_A) = Y_I} \geq 1/k - k^2/n.$$

We highlight that the gap between the probability of $\Gc$ to output $Y_I$ to its probability to output $-Y_I$ is smaller than $1/k$, and that it is not enough to break \DP alone. The point is that we have a bound on the probability of $\Gc$ to output $-Y_I$, and that in \DP, we measure the difference in probabilities using \emph{multiplicative distance}. Thus, when 
\begin{align}\label{eq:overview:boundingK}
\pr{\Gc(I,Y_{-I}, V_A) = Y_I} > e^\eps\cdot \pr{\Gc(I,Y_{-I}, V_A) = -Y_I} 
\end{align}
the algorithm $\Gc$ does imply contradiction to \DP. To make sure that \cref{eq:overview:boundingK} holds, we need to have
\begin{align*}
 1/k - k^2/n > e^\eps\cdot k^2/n,
 \end{align*}
which is the reason we need to  take $k \ll n^{1/3}$.
In the rest of this section, we give a high-level proof sketch of \cref{lemma:property1:prediction:overview}.

\subsection{Proof Overview for  \cref{lemma:property1:prediction:overview}}\label{sec:overview:prediction-lemma}
In this part, we explain the ideas behind the proof of \cref{lemma:property1:prediction:overview}. We start by describing the algorithm $\Gc$. 
\paragraph{The predictor  $\Gc$.} The construction of $\Gc$ leverages the fact that $\Ac$ succeed on a random index $I$  while only seeing a subset of $Y$. Recall that on input $(I, Z_{-I},W)$, the goal of \Gc is to predict the value of $Z_I$. Towards this end, let $Z^b=(Z_1,\dots,Z_{I-1},b,Z_{I+1},\dots,Z_n)$ (hence, $Z\in \set{Z^{-1},Z^{1}}$). Algorithm $\Gc$ computes the following values:\footnote{In the proof itself, $\Gc$ only estimates these values by sampling $R$. For simplicity, here we assume that $\Gc$ can compute them exactly.}
\[\mu^{1} = \eex{R\gets \zn\mid R_I=1}{\Ac(R,Z^1_R,W)}\]
\[\mu^{-1} = \eex{R\gets \zn\mid R_I=1}{\Ac(R,Z^{-1}_R,W}\]
\[\mu^{*} = \eex{R\gets \zn\mid R_I=0}{\Ac(R,Z_R,W)}\]
Importantly, we note that when $R_I=0$,  $Z_R$ does not contain $Z_I$, and thus $\Gc$ can compute $\mu^*$, given its input.
Finally, $\Gc$ outputs $b\in \oo$ if the following two conditions holds:
\[ \size{\mu^b-\mu^*}< 1/4k,\ and,\ \size{\mu^{-b}-\mu^{*}}\ge 1/4k. \] 
Otherwise, $\Gc$ outputs $\bot$.

Namely, $\Gc$ outputs $b$ if the expected value of $\Ac$'s output on $(R, Z_R)$ when $R_I=0$, is closer to the expected value of $\Ac$'s output on $(R, Z^b_R)$ when $R_I=1$ than it is to the expected value of $\Ac$'s output on $(R, Z^{-b}_R)$. We next analyze $\Gc$. 

\paragraph{Analyzing the predictor.} We will show that the following hold:
\begin{enumerate}
    \item\label{overview:item:1} The value of $R_I$ does not change the expectation by much. Formally,
$$ \size{\eex{R\gets \zn\mid R_I=1}{\Ac(R,Z_R, W)}-\eex{R\gets \zn\mid R_I=0}{\Ac(R,Z_R,W)}}<1/4k$$
with  probability at least $1-k^2/n$, while 
\item\label{overview:item:2} Filliping the value of $Z_I$ does change the expectation. Namely, $$\size{\eex{R\gets \zn\mid R_I=1}{\Ac(R,Z^1_R,W)}-\eex{R\gets \zn\mid R_I=1}{\Ac(R,Z^{-1}_R,W)}}> 1/2k$$
with probability at least $1/k$.
\end{enumerate}
Together, the two items above imply that our distinguisher $\Gc$ has the properties we need:
\begin{itemize}
    \item The first item  above implies that $\Gc$ errs and  outputs $-Z_I$  with probability at most $k^2/n$. Indeed, for $b=-Z_I$, it holds that $\mu^{-b}=\eex{R\gets \zn\mid R_I=1}{\Ac(R,Z_R, W)}$.  Thus, \cref{overview:item:1} implies that with all but probability $k^2/n$,  $\size{\mu^{-b}-\mu^{*}}<1/4k$. In this case, by definition of $\Gc$, it does not output $b$.
    
\item In contrast, the two items together imply that $\Gc$ outputs $Z_I$  with probability at least $1/k-k^2/n$. Indeed, for $b=Z_I$,  \cref{overview:item:1} implies that $\size{\mu^b-\mu^*}< 1/4k$ with probability at least $1/k$. Additionally, by the triangular inequality, when both \cref{overview:item:1} and \cref{overview:item:2} hold together we get  that 
$$ \size{\eex{R\gets \zn\mid R_I=0}{\Ac(R,Z_R,W)}-\eex{R\gets \zn\mid R_I=1}{\Ac(R,Z^{-b}_R,W)}}> 1/4k,$$
which by definition  implies that 
$\size{\mu^{-b}-\mu^{*}}\ge 1/4k$ (for $b=Z_I$). By a simple union bound, the probability that both \cref{overview:item:1} and \cref{overview:item:2} hold is at least $1/k - k^2/n$.
\end{itemize}
Therefore, assuming that \cref{overview:item:1} and \cref{overview:item:2} hold, this  conclude the proof of \cref{lemma:property1:prediction:overview}. We next briefly explain why \cref{overview:item:1} and \cref{overview:item:2} hold with the claimed probability.

\paragraph{Bounding the error probability.} We first sketch the proof of \cref{overview:item:1}. We claim that 
\begin{align}\label{eq:overview:side}
    \ppr{Z,W,I}{ \size{\eex{R\gets \zn\mid R_I=0}{\Ac(R,Z_R,W)}-\eex{R\gets \zn\mid R_I=1}{\Ac(R,Z_R,W)}}\ge 1/4k}\le k^2/n.
\end{align}
In the following, we will prove a stronger claim. That is, that for every fixed $Z$ and $W$, it holds that 
\begin{align}\label{eq:overview:equiv}
    \ppr{I}{ \size{\eex{R\gets \zn\mid R_I=0}{\Ac(R,Z_R,W)}-\eex{R\gets \zn\mid R_I=1}{\Ac(R,Z_R,W)}}\ge 1/4k}\le k^2/n.
\end{align}
Fix $Z$ and $W$, and let $F(R)=\Ac(R,Z_R,W)$. Using this notation, we can rewrite \cref{eq:overview:equiv} as
\begin{align}\label{eq:overview:final}
    \ppr{I}{ \size{\eex{R\gets \zn\mid R_I=0}{F(R)}-\eex{R\gets \zn\mid R_I=1}{F(R)}}\ge 1/4k}\le k^2/n.
\end{align}
Finally, \cref{eq:overview:final} follows a more general well-known fact on the uniform distribution over random strings: For a random string $R\gets \zn$, and a random index $I\gets [n]$, the statistical distance between $R|_{R_I=0}$ and $R|_{R_I=1}$ is at most $1/\sqrt{n}$. This fact implies that the function $F$ cannot distinguish too well between samples from $R|_{R_I=0}$ and samples from $R|_{R_I=1}$.\footnote{Importantly, notice that $F$ does not get $I$ as input.} We prove the following stronger lemma that finishes this part of the proof.

\iffull
\begin{lemma}[Informal version of \cref{lem:distance-I}]\label{lem::distance-I:informal}
        Let $R \la \zo^n$. 
	For any (randomized) function $F:\{0,1\}^n\rightarrow \{0,1\}$ and any $k>0$, it holds that
        \begin{align}\label{eq:f-alpha}
            \ppr{i \la [n]}{\size{\:\ex{F(R) \mid R_i = 0}-\ex{F(R) \mid R_i = 1}\:}\geq 1/4k} \leq \Theta\paren{\frac{k^2}{n}},
        \end{align}
        where the expectations are taken over $R$ and the randomness of $F$.
\end{lemma}
\fi

\paragraph{Lower bounding the distinguishing advantage.} Next we sketch the proof of \cref{overview:item:2}. We need to show that 
\begin{align}
\ppr{Z,W,I}{ \size{\eex{R\gets \zn\mid R_I=1}{\Ac(R,Z^1_R,W)}-\eex{R\gets \zn\mid R_I=1}{\Ac(R,Z^{-1}_R,W)}} \ge 1/2k}\ge 1/k.
\end{align}
Or equivalently,
\begin{align}\label{eq:overview:distinguish}
\ppr{Z,W,I}{ \size{\eex{R\gets \zn\mid R_I=1}{\Ac(R,Z_R,W)}-\eex{R\gets \zn\mid R_I=1}{\Ac(R,Z^{(I)}_R,W)}} \ge 1/2k}\ge 1/k.
\end{align}
We claim that \cref{eq:overview:distinguish} follows by the assumption that $\Ac$ is a good distinguisher. Indeed, assume, for the sake of contradiction, that this is not the case. Then, by a simple  computation (and using the fact that $\Ac(R,Z_R,W)\in \set{0,1}$)  we get that 
\begin{align*}
\eex{Z,W,I}{ \size{\eex{R\gets \zn\mid R_I=1}{\Ac(R,Z_R,W)}-\eex{R\gets \zn\mid R_I=1}{\Ac(R,Z^{(I)}_R,W)}} }\le 2/k.
\end{align*}
and thus, using the triangular inequality, we can write,
\begin{align*}
\size{\eex{Z,W,I}{\eex{R\gets \zn\mid R_I=1}{\Ac(R,Z_R,W)}-\eex{R\gets \zn\mid R_I=1}{\Ac(R,Z^{(I)}_R,W)}} }< 2/k.
\end{align*}
Equivalently, we have that,
\begin{align*}
\size{\pr{\Ac(R,Z_R,W)=1\mid R_I=1}-\pr{\Ac(R,Z^{(I)}_R,W)=1\mid R_I=1}}<2/k.
\end{align*}
which means that $\Ac$ fails to distinguish between $(R,Z_R)$ and $(R,Z^{(I)}_R)$ when $R_I=1$. On the other hand, we notice that $\Ac$ also must fail when $R_I=0$. Indeed, when $R_I=0$ it holds that $Z_R=Z^{(I)}_R$, and thus we get  that 
\begin{align*}
\size{\pr{\Ac(R,Z_R,W)=1\mid R_I=0}-\pr{\Ac(R,Z^{(I)}_R,W)=1\mid R_I=0}}=0.
\end{align*}
Finally, since $\Ac$ fails to distinguish $Z_R$ from $Z^{(I)}_R$ both when $R_I=0$ and $R_I=1$, we conclude that $\Ac$ is not a good distinguisher. That is, 
\begin{align*}
\size{\pr{\Ac(R,Z_R,W)=1}-\pr{\Ac(R,Z^{(I)}_R,W)=1}}< 1/k,
\end{align*}
which contradicts our assumption on $\Ac$.

\section{Related Works}\label{sec:related-works}

\textbf{Computational differential privacy.} Computational differential privacy ($\CDP$) can be defined using two primary approaches.
The more flexible and widely used approach is the \emph{indistinguishability-based} definition, which limits the distinguished $g$ in \cref{def:intro:DP}, to those that are computationally efficient. The second approach, known as the \emph{simulation-based} definition, requires that the output distribution of the mechanism $f$ be computationally indistinguishable from that of an (information-theoretic) differentially private mechanism. Various relationships between these and other privacy definitions have been established in \cite{MPRV09} (a more updated survey is provided in \cite{MeisingsethR25}). We remark that our result holds for the weaker definition of indistinguishability-based $\CDP$, which makes it stronger. As a corollary, we also get the equivalence of the power of the definitions for the inner-product task we consider (in our accuracy regime).

\textbf{\CDP in the centralized model.} In the single-party scenario (\ie the centralized model), computational and information-theoretic differential privacy appear to be more closely aligned. Specifically, \cite{GKY11} demonstrated that a broad class of $\CDP$ mechanisms can be converted into an information-theoretic $\DP$ mechanism. On the other hand, \cite{BunCV16,GhaziIK0M23} showed that under certain non-standard and very-strong cryptographic assumptions, there exist somewhat contrived tasks that can be efficiently solved with $\CDP$, but remain infeasible (\cite{BunCV16}) or impossible (\cite{GhaziIK0M23}) under information-theoretic $\DP$. It still remains open whether such separations exist under more standard cryptographic assumptions and for more natural tasks.

\textbf{\CDP in the local model.} At the other end of the spectrum, the \emph{local model} is highly relevant in practical applications. In this setting, each of the (typically many) participants holds a single data element. Protocols achieving information-theoretic $\DP$ in this model often rely on randomized response, which has been shown to be optimal for counting functions, including inner product, as proven by \cite{chan2012optimal}. In contrast, local $\CDP$ protocols can leverage secure multiparty computation to simulate any efficient single-party mechanism, thereby demonstrating a fundamental gap between the power of $\CDP$ and information-theoretic $\DP$. We remark that there exist other approaches to bridge this by relaxing the distributed model, e.g., using a trusted shuffler \cite{BittauEMMRLRKTS17,CheuSUZZ19} (see \cite{CheuSurvey22} for a survey on this model), or partially trusted servers that enable more accurate estimations using weaker and more practical cryptographic tools than secure multi-party computation (e.g., \cite{BangaloreCV24,DamgardKNOP24,BohlerK20}).

\textbf{Two-Party \CDP.} In the two-party (or small-party) setting, the complexity of $\CDP$ protocols is much less clear. Prior works have maintly focused on Boolean functionalities, where each party holds a single sensitive bit, and the objective is to compute a Boolean function over these bits while preserving privacy (\eg XOR). \cite{GMPS13} established that, for any non-trivial Boolean functionality, there is a fundamental gap in accuracy between what can be achieved in the centralized and distributed settings. Moreover, any $\CDP$ protocol that surpasses this accuracy gap would imply the existence of one-way functions. Later, \cite{GKMPS16} demonstrated that an accurate enough $\CDP$ protocol for the XOR function would inherently imply an oblivious transfer protocol. Building on this, \cite{HNOSS20} proved that any meaningful $\eps$-\CDP two-party protocol for XOR necessarily implies an (infinitely-often) key agreement protocol. \cite{HMSS19} refined and extended the results of \cite{GKMPS16,HNOSS20}, showing that any non-trivial $\CDP$ protocol for XOR implies oblivious transfer.

Beyond Boolean functionalities, the complexity of $\CDP$ protocols for more general tasks such as computing low-sensitivity many-bit functions like the inner product remains largely unexplored. The only exceptions are the work of \cite{haitner2016limits}, who applied a generic reduction to the impossibility result of \cite{MPRV09}, concluding that no accurate $\CDP$ protocol for inner product can exist in the \textit{random oracle model}, and more recently, the work of \cite{HaitnerMST22} who showed that any non-trivial $\CDP$ inner-product protocol implies a key-agreement protocol. But none of these results close the gap to $\OT$, and thus, our result provide the first tight characterization (with respect to assumptions) of a natural, non-boolean functionality.

\section{Preliminaries}\label{sec:preliminaries}



\subsection{Distributions and Random Variables}\label{sec:prelim:dist}
The support of a distribution $P$ over a finite set $\cS$ is defined by $\Supp(P) \eqdef \set{x\in \cS: P(x)>0}$. For a distribution or a random variable $D$, let $d\from D$ denote that $d$ was sampled according to $D$. Similarly,  for a set $\cS$, let $x \from \cS$ denote that $x$ is drawn uniformly from $\cS$, and denote by $\cU_{\cS}$ the uniform distribution over $\cS$. For a finite set $\cX$ and a distribution $C_X$ over $\cX$, we use the capital letter $X$ to denote the random variable that takes values in $\cX$ and is sampled according to $C_X$. The {\sf statistical distance} (\aka {\sf~variation distance}) of two distributions $P$ and $Q$ over a discrete domain $\cX$ is defined by $\sdist{P}{Q} \eqdef \max_{\cS\subseteq \cX} \size{P(\cS)-Q(\cS)} = \frac{1}{2} \sum_{x \in \cS}\size{P(x)-Q(x)}$. 
For a vector $x = (x_1,\ldots,x_n)$ and index $i\in [n]$, we let $x_{-i} = (x_1,\ldots,x_{i-1},x_{i+1},\ldots,x_n)$ and $x^{(i)} = (x_1,\ldots,x_{i-1}, -x_i, x_{i+1},\ldots,x_n)$, for a set $\cS \subseteq [n]$ we let $x_{\cS} = (x_i)_{i \in \cS}$ and $x_{-\cS} = (x_i)_{i \in [n]\setminus \cS}$, and for a vector $r \in \zo^n$ we let $x_r = x_{\set{i \colon r_i = 1}}$ and $x_{-r} = x_{\set{i \colon r_i = 0}}$.



\remove{
\subsection{Differential Privacy}\label{sec:prelim:DP}
We use the following standard definition of (information theoretic) differential privacy, due to \citet{DMNS06}. For notational convenience, we focus on databases over $\oo$.
\begin{definition}[Differentially private mechanisms]\label{def:mech}
	A randomized function $f\colon\oo^n\mapsto \zs$ is an {\sf $n$-size, $(\eps,\delta)$-differentially private mechanism} (denoted $(\eps,\delta)$-\DP) if for every neighboring $w,w'\in \oo^n$ and every function $g\colon \zs\mapsto \zo$, it holds that 
	$$
	\pr{g(f(w))=1}\leq e^{\eps}\cdot \pr{g(f(w'))=1} +\delta.
	$$ 	
	If $\delta=0$, we omit it from the notation.
\end{definition}
}

\subsubsection{Computational Differential Privacy}
There are several ways for defining computational differential privacy (see \cref{sec:related-works}). We use the most relaxed version due to \cite{BNO08}. For notational convenience, we focus on databases over $\oo$.
\begin{definition}[Computational differentially private mechanisms]\label{def:ComMech}
	A randomized function ensemble $f=\set{f_\pk\colon\oo^{n(\pk)}\mapsto \zs}$ is an {\sf $n$-size, $(\eps,\delta)$-computationally differentially private} (denoted $(\eps,\delta)$-$\CDP$) if for every poly-size circuit family $\set{\Ac_\pk}_{\pk\in \N}$, the following holds for every large enough $\pk$ and every neighboring $w,w'\in\oo^{n(\pk)}$:
	$$
	\pr{\Ac_\pk(f_\pk(w))=1}\leq e^{\eps(\pk)}\cdot \pr{\Ac_\pk(f_\pk(w'))=1} +\delta(\pk).
	$$ 
	If $\delta(\pk) = \negl(\pk)$, we omit it from the notation. 
\end{definition}

\subsubsection{Two-Party Differential Privacy}\label{sec:DP}
In this section we formally define distributed differential privacy mechanism (\ie protocols). 

\begin{definition}\label{def:DP}
	A two-party protocol $\Pi=(\Ac,\Bc)$ is {\sf $(\eps,\delta)$-differentially private}, denoted $(\eps,\delta)$-$\DP$, if the following holds for every algorithm $\Dc$: let $\V^\Pc(x,y)(\pk)$ be the view of party $\Pc$ in a random execution of $\Pi(x,y)(1^\pk)$. Then for every $\pk,n \in \N$, $x\in \oo^n$ and neighboring $y,y'\in\oo^n$:
	\begin{align*}
	\pr{\Dc(V^\Ac(x,y)(\pk))=1}\le e^{\eps(\pk)}\cdot \pr{\Dc(V^\Ac (x,y')(\pk))=1}+\delta(\pk),
	\end{align*} 
	and for every $y\in \oo^n$ and neighboring $x,x'\in\oo^{n}$:
	\begin{align*}
	\pr{\Dc(V^\Bc(x,y)(\pk))=1}\le e^{\eps(\pk)}\cdot \pr{\Dc(V^\Bc (x',y)(\pk))=1}+\delta(\pk).
	\end{align*} 	
	Protocol $\Pi$ is {\sf $(\eps,\delta)$-computational differentially private}, denoted $(\eps,\delta)$-$\CDP$, if the above inequalities only hold for a non-uniform \ppt $\Dc$ and large enough $\pk$. We omit $\delta = \negl(\pk)$ from the notation. \footnote{Note that define we give for two-party differentially private protocols is a semi-honest definition, in which we ask for the security to hold when the parties interact in an honest execution of the protocol. Since we are proving a lower bound, starting from this weaker guarantee (as opposed to security against malicious players), yields a stronger result.}
\end{definition}

\begin{remark}[The definition for computational differential privacy we use]\label{rem:comDPChannel} 
	An alternative, stronger definition of computational differential privacy, known as simulation-based computational differential privacy, requires that the distribution of each party’s view be computationally indistinguishable from a distribution that ensures privacy in an information-theoretic sense. \cref{def:DP} is a weaker notion in comparison. Consequently, establishing a lower bound for a protocol that satisfies this weaker guarantee (as we do in this work) yields a stronger result.
\end{remark}

\subsection{Useful Claims}

The following two propositions state that given the output of a differentially private function, it is not possible to predict well even a random index (even if all other indexes are leaked). The first proposition handles the information-theoretic case and the second handles the computation case.
\iffull
Both propositions are proven in \cref{sec:missing-proofs:hard-to-guess}. 
\fi

\def\propHardToGuessInf{
    Let $f\colon \oo^n \rightarrow \cY$ be an $(\eps,\delta)$-\DP function, let $g \colon [n] \times \oo^{n-1} \times \cY \rightarrow \set{-1,1,\bot}$ be a (randomized) function, and let $X = (X_1,\ldots,X_n) \la \oo^n$. Then the following holds for every $i \in [n]$ where $X_i^* = g(i,X_{-i},f(X_1,\ldots,X_n))$:
    \begin{align*}
        \pr{X_i^* = X_i} \leq e^{\eps}\cdot \pr{X_i^* = -X_i} + \delta.
    \end{align*}
}

\begin{proposition}\label{prop:hard-to-guess-inf}
    \propHardToGuessInf
\end{proposition}

\def\propHardToGuessComp{
    Let $f = \set{f_{\pk} \colon \oo^{n(\pk)} \rightarrow \zo^{m(\pk)}}_{\pk \in \bbN}$ be an $(\eps,\delta)$-\CDP function ensemble, and let $\set{g_{\pk}}_{\pk \in \bbN}$ be a poly-size circuit family. Then, for large enough $\pk$ and $X = (X_1,\ldots,X_{n(\pk)}) \la \oo^{n(\pk)}$, the following holds for every $i \in [n(\pk)]$ where $X_i^* = g_{\pk}(i,X_{-i},f_{\pk}(X_1,\ldots,X_n))$:
    \begin{align*}
        \pr{X_i^* = X_i} \leq e^{\eps(\pk)}\cdot \pr{X_i^* = -X_i} + \delta(\pk).
    \end{align*}
}

\begin{proposition}\label{prop:hard-to-guess-comp}
    \propHardToGuessComp
\end{proposition}

\remove{
\Enote{Chao's old statement:}
\begin{lemma}\label{lem:distance-I-old}
        Let $R \la \zo^n$. 
	For any function $f:\{0,1\}^n\rightarrow \{0,1\}$ and $\alpha<0.01$, it holds that
	$$
	\Pr_{i\la[n]}\left[\: \size{\:\mathbb{E}[f(R) \mid R_i = 0]-\mathbb{E}[f(R) \mid R_i = 1]\:}\geq \alpha\right]\leq \frac{2+2\log(\frac{1}{\alpha})}{n\alpha^2}.
	$$
\end{lemma}
\begin{proof}
	Define $S_1=\{r \in \zo^n \colon f(r)=1\}$. Then for any $i\in[n]$, we have
	$$
	\begin{array}{rl}
		\size{\mathbb{E}[f(R) \mid R_i = 0]-\mathbb{E}[f(R) \mid R_i = 1]}
		&=\size{\Pr[R\in S_1|R_i=0]-\Pr[R\in S_1|R_i=1]}\\
		&=\size{\frac{\Pr[R_i=0|R\in S_1]\cdot\Pr[R\in S_1]}{\Pr[R_i=0]}-\frac{\Pr[R_i=1|R\in S_1]\cdot\Pr[R\in S_1]}{\Pr[R_i=1]}}\\
		&=\frac{2\size{S_1}}{2^n}\size{\Pr[R_i=0|R\in S_1]-\Pr[R_i=1|R\in S_1]}
	\end{array}
	$$
	When $|S_1|\leq \alpha\cdot 2^{n-1}$, we have $\size{\mathbb{E}[f(R) \mid R_i = 0]-\mathbb{E}[f(R) \mid R_i = 1]}\leq\frac{2\size{S_1}}{2^n}\leq \alpha$ for any $i\in[n]$. Hence, in the following, we assume $|S_1|> \alpha\cdot 2^{n-1}$.

    Define $I_{bad}=\{i \colon \size{\Pr[R_i=0|R\in S_1]-\Pr[R_i=1|R\in S_1]} \geq 2\alpha\}$ and $k=\size{I_{bad}}$, and denote $I_{bad}=\{i_1,\dots,i_k\}$. Define $(X_{i_1}, \ldots X_{i_k}) = (R_{i_1},\dots,R_{i_k})\mid_{R \in S_1}$. 
    Consider the min-entropy
	$$
	\begin{array}{rl}
		H_{min}(X_{i_1},\dots,X_{i_k})&\leq H(X_{i_1},\dots,X_{i_k})\\
		&\leq \sum_{j=1}^k H(X_{i_j})\\
		&\leq k\cdot \left(-(\frac{1}{2}+2\alpha)\cdot\log(\frac{1}{2}+2\alpha)-(\frac{1}{2}-2\alpha)\cdot\log(\frac{1}{2}-2\alpha)\right)\\
            &=k\cdot \left(-(\frac{1}{2}+2\alpha)\cdot(\log(1+4\alpha)-1)-(\frac{1}{2}-2\alpha)\cdot(\log(1-4\alpha)-1)\right)\\
            &=k\cdot \left(1-(\frac{1}{2}+2\alpha)\cdot\log(1+4\alpha)-(\frac{1}{2}-2\alpha)\cdot\log(1-4\alpha)\right),
		
	\end{array}
	$$
	where $H_{min}(Y)$ is the minimum entropy of $Y$ and $H(Y)$ is the Shannon entropy of $Y$.\Enote{add to preliminaries.}
        The third inequality holds since by the definition of $I_{bad}$, for every $j \in [k]$ it holds that $\size{\pr{X_{i_j} = 1}-\pr{X_{i_j} = 0}} > 2\alpha$, and therefore $H(X_{i_j}) \leq H(1/2 + 2\alpha)$\Enote{define}.
	
	Therefore, there exists $b_1,\dots,b_k\in\{0,1\}$, such that 
	
	\begin{align}\label{eq:min-entropy-result}
		\Pr\left[(R_{i_1},\ldots,R_{i_k}) = (b_1,\ldots,b_k) \mid R\in S_1\right]
		&= \pr{(X_{i_1},\ldots,X_{i_k}) = (b_1,\ldots,b_k)}\\
		&= 2^{-H_{min}(X_{i_1},\dots,X_{i_k})}\nonumber\\
		&\geq 2^{k\cdot \left(-1+(\frac{1}{2}+2\alpha)\cdot\log(1+4\alpha)+(\frac{1}{2}-2\alpha)\cdot\log(1-4\alpha)\right)}.\nonumber
	\end{align}
	
	Let $S_{bad}=\{r \in \zo^n  \colon \set{(r_{i_1},\ldots,r_{i_k}) = (b_1,\ldots,b_k)} \land \set{r\in S_1}\}$.
	It holds that
	\begin{align*}
		|S_{bad}|
		&= \size{S_1} \cdot \Pr\left[(R_{i_1},\ldots,R_{i_k}) = (b_1,\ldots,b_k) \mid R\in S_1\right]\\
		&\geq \alpha\cdot 2^{n-1}\cdot2^{k\cdot \left(-1+(\frac{1}{2}+2\alpha)\cdot\log(1+4\alpha)+(\frac{1}{2}-2\alpha)\cdot\log(1-4\alpha)\right)},
	\end{align*} 
	where the inequality holds by \cref{eq:min-entropy-result} and since $\size{S_1} \geq \alpha\cdot 2^{n-1}$.
	Notice that any string in $S_{bad}$ depends on at most $n-k$ bits. It implies that $|S_{bad}|\leq 2^{n-k}$. Therefore, we have
	$$
	\begin{array}{rl}
		&2^{n-k}\geq \alpha\cdot 2^{n-1}\cdot2^{k\cdot \left(-1+(\frac{1}{2}+2\alpha)\cdot\log(1+4\alpha)+(\frac{1}{2}-2\alpha)\cdot\log(1-4\alpha)\right)} \\
		\Rightarrow& n-k \geq \log \alpha+n-1+k\cdot \left(-1+(\frac{1}{2}+2\alpha)\cdot\log(1+4\alpha)+(\frac{1}{2}-2\alpha)\cdot\log(1-4\alpha)\right)\\
		\Rightarrow& 1-\log \alpha \geq k\cdot((\frac{1}{2}+2\alpha)\cdot\log(1+4\alpha)+(\frac{1}{2}-2\alpha)\cdot\log(1-4\alpha))\\
		\Rightarrow& 1-\log \alpha \geq k\cdot(4\alpha\cdot\log(1+4\alpha)+(\frac{1}{2}-2\alpha)\cdot\log(1-16\alpha^2))\\
        \Rightarrow& 1-\log\alpha \geq k\cdot(15.9\alpha^2-8\alpha^2+32\alpha^3)=k\cdot(7.9\alpha^2+32\alpha^3)>0.5k\alpha^2\\
		\Rightarrow& k\leq \frac{2-2\log \alpha}{\alpha^2} = \frac{2+2\log (1/\alpha)}{\alpha^2},
	\end{array}
	$$
	Where the third transition holds since 
	\begin{align*}
		\lefteqn{(\frac{1}{2}+2\alpha)\cdot\log(1+4\alpha)+(\frac{1}{2}-2\alpha)\cdot\log(1-4\alpha)}\\
		&= 4\alpha\cdot\log(1+4\alpha) + (\frac{1}{2}-2\alpha)\paren{\log(1+4\alpha)+\log(1-4\alpha)}\\
		&= 4\alpha\cdot\log(1+4\alpha)+(\frac{1}{2}-2\alpha)\cdot\log(1-16\alpha^2),
	\end{align*}
	and the forth transition holds since $4\alpha\cdot\log(1+4\alpha)+(\frac{1}{2}-2\alpha)\cdot\log(1-16\alpha^2) > 15.9\alpha^2-8\alpha^2+32\alpha^3$ for $\alpha < 0.01$.
	Thus, we conclude that 
	$$
	\Pr_{i\la[n]}\left[\size{\mathbb{E}[f(R) \mid R_i=0]-\mathbb{E}[f(R) \mid R_i = 1]}\geq \alpha\right]\leq \frac{k}{n}\leq \frac{2+2\log (1/\alpha)}{n\alpha^2}.
	$$
\end{proof}
}

\subsection{Channels and Two-Party Protocols}\label{sec:protocol}

\paragraph{Channels.}A channel is simply a distribution of a pair of tuples defined as follows. 
\begin{definition}[Channels]\label{def:channel} A {\sf channel} $C_{(X,U)(Y,V)}$ of size $\isize$ over alphabet $\Sigma$ is a probability distribution over $(\Sigma^\isize \times\zo^\ast) \times(\Sigma^\isize \times\zo^\ast)$. The ensemble $C_{(X,U)(Y,V)}= \set{C_{(X_\pk,U_\pk)(Y_\pk,V_\pk)}}_{\pk\in \N}$ is an $\isize$-size channel ensemble, if for every $\pk\in \N$, $C_{(X_\pk,U_\pk)(Y_\pk,V_\pk)}$ is an $\isize(\pk)$-size channel. 
We refer to $X$ and $Y$ as the {\sf local outputs}, and to $U$ and $V$ as the {\sf views}.	
\end{definition}

We view a  channel as the experiment in which there are two parties $\Ac$ and $\Bc$.  Party $\Ac$ receives ``output'' $X$ and ``view'' $U$, and party $\Bc$ receives ``output'' $Y$ and ``view'' $V$. Unless stated otherwise, the channels we consider are over the alphabet $\Sigma = \oo$. We naturally identify channels with the distribution that characterizes their output.

\subsubsection{Two-Party Protocols}

A two-party protocol $\Pi=(\Ac,\Bc)$ is \ppt if the running time of both parties is polynomial in their input length. We let $\Pi(x,y)(z)$ or $(\Ac(x),\Bc(y))(z)$ denote a random execution of $\Pi$ on a common input $z$, and private inputs $x,y$.

\begin{definition}[Oracle-aided protocols]\label{def:ChannelAidedProtocol}
	In a two-party protocol $\Pi$ with oracle access to a {\sf protocol} $\Psi$, denoted $\Pi^\Psi$, the parties make use of the \textit{next-message function} of $\Psi$.\footnote{The function that on a partial view of one of the parties, returns its next message.} In a two-party protocol $\Pi$ with oracle access to a {\sf channel} $C_{Z W}$, denoted $\Pi^C$, the parties can jointly invoke $C$ for several times. In each call, an independent pair $(z,w)$ is sampled according to $C_{Z W}$, one party gets $z$, the other gets $w$.
\end{definition}

\begin{definition}[The channel of a protocol]\label{def:ChannlOfProtocol}
	For a no-input two-party protocol $\Pi= (\Ac,\Bc)$, we associate the channel $C_\Pi$, defined by $\C_\Pi= C_{(X, U),(Y, V)}$, where $X$ and $Y$ are the local outputs of $\Ac$ and $\Bc$ (respectively) and
	$U$ and $V$ are the local views of $\Ac$ and $\Bc$ (respectively).
    
	For a two-party protocol $\Pi$ that gets a security parameter $1^\pk$ as its (only, common) input, we associate the channel ensemble $ \set{C_{\Pi(1^\pk)}}_{\pk\in \N}$. 
\end{definition}

\begin{definition}[$(\alpha,\gamma)$-Accurate channel]\label{def:accurate-func}
	A channel $C = C_{(X, U),(Y, V)}$ is {\sf $(\alpha,\gamma)$-accurate for the function $f$}, if $\ppr{C}{\size{\out(V)-f(X,Y)}\leq \alpha}\ge \gamma$, where $\out(V)$ is the designated output.
    A channel ensemble $C_{(X, U),(Y, V)}= \set{C_{(X_\pk, U_\pk),(Y_\pk, V_\pk)}}_{\pk\in \N}$ is  $(\alpha,\gamma)$-accurate for  $f$ if $C_{(X_\pk, U_\pk),(Y_\pk, V_\pk)}$ is $(\alpha(\pk),\gamma(\pk))$-accurate for $f$, for every $\pk \in \N$.
\end{definition}

\subsubsection{Differentially Private Channels}\label{sec:DPChannel}
Differentially private channels are naturally defined as follows:
\begin{definition}[Differentially private channels]\label{def:DPChannel}
	An $n$-size channel $C = C_{(X, U),(Y, V)}$ with $X, Y$ over $\oo^n$ 
	is {\sf$(\eps,\delta)$-differentially private} (denoted $(\eps,\delta)$-$\DP$) if for every $x \in \Supp(X)$ there exists an $n$-size $(\eps,\delta)$-$\DP$ mechanisms $\Mc_x$ such that $(X,Y,U) \equiv (X,Y,\Mc_X(Y))$, and for every $y \in \Supp(Y)$ there exists an $n$-size $(\eps,\delta)$-$\DP$ mechanisms $\Mc_y'$ such that $(X,Y,V) \equiv (X,Y,\Mc_Y'(X))$. In addition, we say that the channel is \emph{uniform} if $X$ and $Y$ are independent random variables uniformly distributed in $\oo^n$. 
\end{definition}

\begin{definition}[Computational differentially private channels]\label{def:CDPChannel}
	An $n$-size channel ensemble $C = \set{C_{(X_\pk, U_\pk),(Y_\pk, V_\pk)}}_{\pk\in\N}$ with $X_\pk, Y_\pk$ over $\oo^n$ 
	is {\sf$(\eps,\delta)$-computationally differentially private} (denoted $(\eps,\delta)$-$\CDP$) if for every ensemble $\set{x_\pk \in \Supp(X_\pk)}_{\pk\in\N}$ there exists an $n$-size $(\eps,\delta)$-\CDP mechanisms ensemble $\set{\Mc_{x_\pk}}_{\pk\in\N}$ such that $(X_\pk,Y_\pk,U_\pk) \equiv (X_\pk,Y_\pk,\Mc_{X_\pk}(Y_\pk))$, for every $\pk\in\N$, and for every ensemble $\set{y_\pk \in \Supp(Y_\pk)}_{\pk\in\N}$ there exists an $n$-size $(\eps,\delta)$-$\CDP$ mechanisms ensemble $\set{\Mc'_{y_\pk}}_{\pk\in\N}$ such that $(X_\pk,Y_\pk,V_\pk) \equiv (X_\pk,Y_\pk,\Mc_{Y_\pk}'(X_\pk))$ for every $\pk\in \N$. In addition, we say that the channel is \emph{uniform} if $X_\pk$ and $Y_\pk$ are independent random variables uniformly distributed in $\{\pm 1\}^n$ for all $\pk\in\N$.
\end{definition}


\subsubsection{Weak Erasure Channel (\WEC)}

\begin{definition}[\WEC]\label{def:WEC}
	A channel $((O_A,V_A), (O_B,V_B))$ with $O_A \in \set{0,1}$ and $O_B \in \set{0,1,\bot}$ is a {\sf weak erasure channel}, denoted $(\alpha,p,q)$-$\WEC$, if:
	\begin{itemize}
		\item Random erasure: $\pr{O_B = \perp} = 1/2$.
		
		\item Agreement: $\pr{O_A\ne O_B\mid O_B\ne \bot}\le \alpha$.
		
		\item Secrecy:
		
		\begin{enumerate}
			\item For every algorithm $\Dc$ it holds that\label{WEC:item:A}
			\begin{align*}
				\size{\pr{\Dc(V_A) = 1 \mid O_B \neq \perp} - \pr{\Dc(V_A) = 1 \mid O_B = \perp}} \le p
			\end{align*}
			(Alice doesn't know if $O_B = \perp$.)
			
			\item For every algorithm $\Dc$ it holds that\label{WEC:item:B}
			\begin{align*}
				\pr{\Dc(V_B) = O_A \mid O_B=\bot} \leq \frac{1+q}{2}.
			\end{align*}
			(i.e., if $O_B=\bot$, Bob don't know what is the value of $O_A$).
			
			
		\end{enumerate}
	\end{itemize}
   We say that a channel ensemble $C=\set{C_\pk}_{\pk\in N}$ is a {\sf computational weak erasure channel}, denoted $(\alpha,p,q)$-\CompWEC, if for every \ppt algorithm $\Dc$ and every sufficiently large $\pk\in\N$, $C_\pk$ satisfies the properties stated in the items above, where the secrecy property holds with respect to a \ppt algorithm $\Dc$. A protocol $\Lambda$ is said to be $(\alpha,p,q)$-$\CompWEC$, if the ensemble induces by the protocol (that is, $C=\set{C_{\Lambda(\pk)}}_{\pk\in\N}$) is $(\alpha,p,q)$-$\CompWEC$.  
\end{definition}

\subsubsection{Approximate Weak Erasure Channel (\AWEC)}\label{sec:AWEC}

\begin{definition}[\AWEC]\label{def:AWEC}
	A channel $C = ((O_A,V_A), (O_B,V_B))$ over $([-n,n] \times \zo^*) \times (([-n,n] \cup \bot)  \times \zo^*)$ is an {\sf approximate weak erasure channel}, denoted $(\ell,\alpha,p,q)$-\AWEC if:
	\begin{itemize}
		
		\item Random erasure: $\pr{O_B = \perp} = 1/2$.
		
		\item Accuracy: $\pr{\size{O_A - O_B} > \ell \mid O_B \ne \bot}\le \alpha$.
		
		\item Secrecy:
		
		\begin{enumerate}
			\item For every algorithm $\Dc$ it holds that\label{AWEC:item:A}
			\begin{align*}
				\size{\pr{\Dc(V_A) = 1 \mid O_B \neq \perp} - \pr{\Dc(V_A) = 1 \mid O_B = \perp}} \le p
			\end{align*}
			(Alice doesn't know if $O_B=\bot$).
			
			\item For every algorithm $\Dc$ it holds that\label{AWEC:item:B}
			\begin{align*}
				\pr{\size{\Dc(V_B) - O_A} \leq 1000 \ell \mid O_B=\bot} \leq q.
			\end{align*}
			(i.e., if $O_B=\bot$, Bob can't estimate the value of $O_A$ with error $\leq 1000 \ell$).
		\end{enumerate}
	\end{itemize}
     We say that a channel ensemble $C=\set{C_\pk}_{\pk\in N}$ is a {\sf computational approximate weak erasure channel}, denoted $(\ell,\alpha,p,q)$-\CompAWEC, if for every \ppt algorithm $\Dc$ and every sufficiently large $\pk\in\N$, $C_\pk$ satisfies the properties stated in the items above. A protocol $\Gamma$ is said to be $(\ell,\alpha,p,q)$-$\CompAWEC$, if the ensemble induced by the protocol (that is, $C=\set{C_{\Gamma(\pk)}}_{\pk\in\N}$) is $(\ell,\alpha,p,q)$-$\CompAWEC$.  
\end{definition}

We will make use of the following lemma, which shows that for some choices of the parameters, \AWEC implies \WEC. The lemma is proven in \cref{sec:AWEC-to-WEC}.

\begin{lemma}\label{lemma:AWEC-to-WEC}
	For every $\ell> 0$, there exists a \ppt protocol $\Lambda = (\Pc_1,\Pc_2)$ such that given an oracle access to an $(\ell,\alpha,p,q)$-\AWEC $C$, the channel $\tilde{C}$ induced by $\Lambda^C$ is $(\alpha'=\alpha+0.001,\: p' = p ,\:  q' = 1/2 + 2(q+0.01))$-\WEC.
	Furthermore, the proof is constructive in a black-box manner:
	\begin{enumerate}
		\item There exists an oracle-aided \ppt algorithm $\Ec_1$ such that for every channel $C = ((\OA,\VA), (\OB,\VB))$ and algorithm $\Dc$ violating the \WEC secrecy property~\ref{WEC:item:A} of $\tilde{C}$, algorithm $\Ec_1^{\Dc}$ violates the \AWEC secrecy property~\ref{AWEC:item:A} of $C$.
		
		\item There exists an oracle-aided \ppt algorithm $\Ec_2$ such that for every channel $C = ((\OA,\VA), (\OB,\VB))$ and algorithm $\Dc$ violating the \WEC secrecy property~\ref{WEC:item:B} of $\tilde{C}$, algorithm $\Ec_2^{\Dc}$ violates the \AWEC secrecy property~\ref{AWEC:item:B} of $C$.
	\end{enumerate}
\end{lemma}

Since \cref{lemma:AWEC-to-WEC} is constructive, the following is an immediate corollary.
\begin{corollary}\label{cor:CompAWEC to CompWEC}
There exists an oracle aided \ppt protocol $\Lambda$, such that given a protocol $\Gamma$ that induces $(\ell,\alpha,p,q)$-\CompAWEC, it holds that $\Lambda^\Gamma$ is $(\alpha'=\alpha+0.001,\: p' = p ,\:  q' = 1/2 + 2(q+0.01))$-\CompWEC.  
\end{corollary}
\begin{proof}[Proof of \ref{cor:CompAWEC to CompWEC}]
Let $\Lambda$ be the \ppt algorithm guaranteed  by Lemma \ref{lemma:AWEC-to-WEC}. Given an $(\ell,\alpha,p,q)$-\CompAWEC protocol $\Gamma$, we define $\Lambda(\pk)={\Lambda^{\Gamma(\pk)}(\pk)}$. Assume towards a contradiction that $\Lambda$ is not a $(\alpha',p',q')$-\CompWEC. It follows that there exists a \ppt $\Dc$ that for infinity many $\pk\in\N$ contradicts one of the \WEC secrecy properties of channel ensemble $\set{C_{\Lambda(\pk)}}_{\pk\in\N}$. Fix $\pk\in\N$ for which this holds. By Lemma \ref{lemma:AWEC-to-WEC}, there exists a \ppt $\Ec^\Dc$ that for every such $\pk$  contradicts one of the secrecy properties of the channel $C_{\Gamma(\pk)}$. This implies that for infinity many $\pk\in\N$, $\Ec^\Dc$  contradict the secrecy of the channel ensemble $\set{C_{\Gamma(\pk)}}_{\pk\in\N}$, which is a contradiction since this would means that $\Gamma$ is not a $(\ell,\alpha,p,q)$-\CompAWEC.       
\end{proof}

\subsection{Oblivious Transfer (\OT)}

\paragraph{Secure Computation.}
We use the standard notion of securely computing a functionality, \cf  \cite{Goldreich04}.
\begin{definition}[Secure computation]\label{def:SFE}
	A two-party protocol {\sf securely computes a functionality $f$}, if it does so according to the real/ideal paradigm.   We add the term perfectly/statistically/computationally/non-uniform computationally, if the simulator's output is  perfect/statistical/computationally indistinguishable/  non-uniformly indistinguishable from  the real distribution.  The protocol have the above notions of security {\sf against semi-honest  adversaries}, if its security only  guaranteed to holds against an adversary that follows the prescribed protocol.   Finally, for the case of perfectly secure computation, we naturally apply the above notion also to the non-asymptotic case: the protocol with no security parameter perfectly  compute a functionality $f$.
	
	A two-party protocol {\sf securely computes a functionality ensemble $f$ with oracle to a channel $C$}, if it does so according to the above definition when the parties have access to a trusted party computing $C$. All the above adjectives naturally extend to this setting.
\end{definition}

\paragraph{Oblivious Transfer.}
The (one-out-of-two) oblivious transfer functionality is defined as follows.
\begin{definition}[oblivious transfer functionality $f_{\OT}$]\label{def:OTfunc}
	The oblivious transfer functionality over $\zo \times (\zs)^2$ is defined by  $f_{\OT} (i,(\sigma_0,\sigma_1)) = (\perp,\sigma_i)$.
\end{definition}
A protocol is $\ast$ secure OT,   for \\$\ast\in \set{\text{semi-honest statistically/computationally/computationally non-uniform}}$, if it  compute the $f_{\OT}$  functionality with $\ast$ security.

    
    

We make use of the following useful results by Wullschleger on oblivious transfer amplification from weak channels.
\begin{theorem}[\cite{Wullschleger09}, from \WEC to statistically secure \OT]\label{thm:WEC TO OT IT}
    There exists an oracle aided protocol $\Pi$ such that the following holds: Given a $(\alpha,p,q)$-\WEC $C$, if $44(\alpha+p)\le 1-q$ then $\Pi^{C}(1^\pk)$ is a semi-honest statistically secure \OT.
\end{theorem}

The following computational version of \cref{thm:WEC TO OT IT} is implicit in \cite{Wullschleger09} and is based on the computational proof explicitly stated in \cite{Wul07} (see Section 6 in \cite{Wullschleger09} for discussion).   

\begin{theorem}[\cite{Wullschleger09,   Wul07}, from \CompWEC to computinally secure \OT]\label{thm:WEC TO OT Comp}
    There exists an oracle aided protocol $\Pi$ such that the following holds: Given a $(\alpha,p,q)$-\CompWEC protocol $\Lambda$, if $44(\alpha+p)\le 1-q$ then $\Pi^{\Lambda}$ is a semi-honest computational secure \OT.
\end{theorem}


    
    
\section{Oblivious Transfer from $\DP$ Inner Product}\label{sec:main_theorems}

\iffull

In this section, we state our main theorems. We first state and prove the information theoretic case in \cref{sec:DP TO OT IT}, and then prove the computational case in \cref{sec:DP TO OT Comp}.

\fi

\subsection{Information-Theoretic Case}\label{sec:DP TO OT IT}
The following is our main theorem (for the information theoretic case) which shows that given a sufficiently accurate and private $\DP$ channel, we can construct a statistically secure semi-honest oblivious transfer protocol.

\begin{theorem}\label{thm:DPIP-to-OT}
There exist constants $c_1,c_2>0$ and an oracle-aided \ppt protocol $\Pi$ such that the following holds for large enough $n \in \bbN$ and for 
$\eps \leq \log^{0.9} n$, $\delta \leq \frac1{3n}$, and $\ell \leq e^{-c_1  \eps}  c_2\cdot n^{1/6}$:
    Let $C = ((X,U),(Y,V))$ be an $(\eps,\delta)$-\DP channel with independent $X,Y \la \oo^n$ that is $(\ell,0.999)$-accurate for the inner-product functionality (i.e., $\ppr{C}{\size{O-\ip{X,Y}} \leq \ell} \geq 0.999$).
    Then $\Pi^C$ is a semi-honest statistically secure $\OT$ protocol.
\end{theorem}

To prove Theorem \ref{thm:DPIP-to-OT}, we will make use of the following technical Lemma.

\begin{lemma}\label{lemma:DPIP-to-AWEC}
        There exists a \ppt protocol $\Gamma=(\Ac,\Bc)$ such that for every $c_1,c_2,n,\eps,\delta,\ell$ and $C$ as in \cref{thm:DPIP-to-OT},
        the channel $\tilde{C}$ induced by the execution of $\Gamma^C$ is $(\ell, \alpha=0.001, p=0.001, q=0.001)$-\AWEC (\cref{def:AWEC}).
\remove{
	There exist constant $c_1,c_2 > 0$ and an oracle-aided \ppt protocol $\Gamma=(\Ac,\Bc)$ such that the following holds for large enough $n \in \bbN$ and for 
	$\eps \leq \log^{0.9} n$, $\delta \leq \frac1{3n}$, and $\ell \leq e^{-c_1  \eps}  c_2\cdot n^{1/6}$:
	Let $C = ((X,U),(Y,V))$ be an $(\eps,\delta)$-\DP channel with independent $X,Y \la \oo^n$ that is $(\ell,0.999)$-accurate for the inner-product functionality (i.e., $\ppr{C}{\size{O-\ip{X,Y}} \leq \ell} \geq 0.999$). Then the channel $\tilde{C}$ induced by the execution of $\Gamma^C$ is $(\ell, \alpha=0.001, p=0.001, q=0.001)$-\AWEC (\cref{def:AWEC}). 
}
    Furthermore, the proof is constructive in a black-box manner:
	\begin{enumerate}
		\item There exists an oracle-aided \ppt algorithm $\Act$ such that for every channel $C = ((X,U),(Y,V))$ and algorithm $\Ac$ violating the \AWEC secrecy property~\ref{AWEC:item:A} of $\tilde{C}$ (the channel of $\Gamma^C$), the following holds for $Y^*_i = \Act^{\Ac}(i,\: Y_{-i}, \: X, \: U)$:\label{item:privacy-of-Y}
		\begin{align*}
			\eex{i \la [n]}{\pr{Y^*_i = Y_i }} > e^{\eps} \cdot \eex{i \la [n]}{\pr{Y^*_i = -Y_i }} + \delta.
		\end{align*}
		
		

		\item There exists an oracle-aided \ppt algorithm $\Bct$ such that for every channel $C = ((X,U),(Y,V))$ and algorithm $\Bc$ violating the \AWEC secrecy property~\ref{AWEC:item:B} of $\tilde{C}$ (the channel of $\Gamma^C$),  the following holds for $X^*_i = \Bct^{\Bc, C}(i,\: X_{-i}, \: Y, \: V)$:\label{item:privacy-of-X}
		\begin{align*}
			\eex{i \la [n]}{\pr{X^*_i = X_i }} > e^{\eps} \cdot \eex{i \la [n]}{\pr{X^*_i = -X_i }} + \delta.
		\end{align*} 
	\end{enumerate}
\end{lemma}

The proof of \cref{lemma:DPIP-to-AWEC} is given in \cref{sec:DPIP_to_WAEC}. But first we use \cref{lemma:DPIP-to-AWEC} to prove \cref{thm:DPIP-to-OT}. 

\begin{proof}[Proof of \cref{thm:DPIP-to-OT}]
By \cref{lemma:DPIP-to-AWEC} there exists a \ppt protocol $\Gamma\coloneqq\Gamma^C$ such that channel $\tilde{C}\coloneqq C_{\Gamma(\pk)}$ is a $(\ell, \alpha=0.001, p=0.001, q=0.001)$-\AWEC. By \cref{lemma:AWEC-to-WEC} (and \cref{prop:hard-to-guess-inf}) there exists a \ppt protocol $\Lambda\coloneqq\Lambda^{\tilde{C}}$ such that the channel $\hhC=C_{\Lambda(\pk)}$ is a $(\alpha'=0.002,\: p' = 0.001 ,\:  q' = 1/2 + 0.022))$-\WEC. Since $44(\alpha'+p')<1-q'$, by  \cref{thm:WEC TO OT IT} there exists a \ppt protocol $\Pi$, such that $\Pi^{\hhC}$ is a semi-honest statistically secure \OT, concluding the proof.
\end{proof}

\iffull

\subsection{Computational Case}\label{sec:DP TO OT Comp}

In this section, we state and prove our results for the computational case. We show that a \CDP (computational differential private) protocol that estimates the inner product well implies a semi-honest computationally secure oblivious transfer protocol.

\begin{theorem}[Restatement of \cref{thm:intro:main}]\label{thm:CDPIP-to-OT}
There exist constant $c_1,c_2 > 0$ and an oracle-aided \ppt protocol $\Pi$ such that the following holds for large enough $n \in \bbN$ and for 
$\eps \leq \log^{0.9} n$, $\delta \leq \frac1{3n}$, and $\ell \leq e^{-c_1  \eps}  c_2\cdot n^{1/6}$:
    Let $\Psi$ be an $(\eps,\delta)$-\CDP protocol that is $(\ell,0.999)$-accurate for the inner-product functionality. 
    Then $\Pi^\Psi$ is a semi-honest computationally secure oblivious transfer protocol.
\end{theorem}

By the result of \cite{GoldreichMW87}, in the computational setting, we can ``compile" any semi-honest computational oblivious transfer protocol into a protocol that is secure against any \ppt (malicious) adversary (assuming one-way functions). We state this formally in the following corollary.

\begin{corollary}\label{cor:mal OT}
   Let $\eps,\delta,\ell$ be as in \cref{thm:CDPIP-to-OT}. If there exists a protocol $\Psi$ that is $(\eps,\delta)$-\CDP  and is $(\ell,0.999)$-accurate for the inner-product functionality, then there exists a computationally secure oblivious transfer protocol. 
\end{corollary}
\begin{proof}[Proof of \cref{cor:mal OT}]
By \cref{thm:CDPIP-to-OT}, there exists a semi-honest computational secure oblivious transfer protocol $\Pi$. Note that by \cite{ImpagliazzoLu89}, $\Pi$ implies the existence of one-way functions and by \cite{GoldreichMW87}, using the one-way function, we can compile $\Pi$ into a (computational) \OT secure against arbitrary adversaries.
\end{proof}

\paragraph{Proof of \cref{thm:CDPIP-to-OT}.}

 In order to use \cref{lemma:DPIP-to-AWEC}, and similar to \cite{HaitnerMST22}, we first convert the protocol $\Psi$ into a (no input) protocol such that the \CDP-channel it induces, is uniform and accurately estimates the inner-product functionality. Such a transformation is simply the following protocol that invokes $\Psi$ over uniform inputs, and each party locally outputs its input. 
\begin{protocol}[$\hPsi^{\Psi} = (\hAc,\hBc)$]\label{prot:EDPtoSV}
	\item Common input: $1^\kappa$.
	\item Operation:
	\begin{enumerate}
		
		\item $\hAc$ samples $x \gets \oo^{\pn(\kappa)}$ and $\hBc$ samples $y\gets \oo^{\pn(\kappa)}$. 
		
		\item The parties interact in \remove{a random execution protocol }$\Psi(1^\kappa)$, with $\hAc$ playing the role of $\Ac$ with private input $x$, and $\hBc$ playing the role of $\Bc$ with private input $y$.
		
		\item $\hAc$ locally outputs $x$ and $\hBc$ locally outputs $y$. 
	\end{enumerate}
\end{protocol}

Let $C$ be the channel ensemble induced by $\hPsi$, letting its designated output (the function $\out$) be the designated output of the embedded execution of $\Psi$. The following fact is immediate by definition.

\begin{proposition}\label{prop:EDP to SV}
	The channel ensemble $C$ is $(\eps,\delta)$-$\CDP$, and has the same accuracy for computing the inner product as protocol $\Psi$ has.
\end{proposition}

We first prove the computational version of \cref{lemma:DPIP-to-AWEC}.
\begin{claim}\label{claim:Comp DP to AWEC}
    There exist constant $c_1,c_2 > 0$ and an oracle-aided \ppt protocol $\Gamma=(\Ac,\Bc)$ such that the following holds for large enough $n \in \bbN$ and for 
	$\eps \leq \log^{0.9} n$, $\delta \leq \frac1{3n}$, and $\ell \leq e^{-c_1  \eps}  c_2\cdot n^{1/6}$:
    Let $\Psi$ be an $(\eps,\delta)$-\CDP protocol that is $(\ell,0.999)$-accurate for the inner-product functionality. Then the protocol $\Gamma^{\Psi}$ is an  $(\ell, \alpha=0.001, p=0.001, q=0.001)$-\CompAWEC protocol.
\end{claim}
\begin{proof}[Proof of \cref{claim:Comp DP to AWEC}]
Recall that $\Psi$ is an $(\eps,\delta)$-$\DP$ protocol that is $(\ell,0.999)$-acccurate for the inner-product functionality for 
	$\eps \leq \log^{0.9} n$, $\delta \leq \frac1{3n}$, and $\ell \leq e^{-c_1  \eps}  c_2\cdot n^{1/6}$. By \cref{prop:EDP to SV}, $\hPsi\coloneqq\hPsi^\Psi$ induces a channel ensemble $C=\set{C_\pk=((X_\pk,U_\pk),(Y_\pk,V_\pk))}_{\pk\in\N}$ that is $(\eps,\delta)$-\CDP and $(\ell,0.999)$-acccurate for the inner-product functionality. Let $\Gamma$ be the \ppt protocol guarrented by \cref{lemma:DPIP-to-AWEC}, we now claim that $\Gamma^{\Psi}(\pk)\coloneqq\Gamma^{C_\pk}(\pk)$ is an $(\ell, \alpha=0.001, p=0.001, q=0.001)$-\CompAWEC protocol. 
Assume towards contradiction that this does not hold, then by \cref{lemma:DPIP-to-AWEC}, there exists a \ppt $\Dc$, such that for infinity many $\pk\in\N$, $\Dc$ contradicts one of the two secrecy properties of $\CompAWEC$. Fix such $\pk\in\N$ and omit it from the notation when clear from the context, and without loss of generality, assume that $\Dc$ contradicts the first secrecy property of $\CompAWEC$ (the case where $D$ contradicts the second property is essentially identical). By the first item of \cref{lemma:DPIP-to-AWEC}, there exists a \ppt algorithm $\tilde{\Ac}$ such that 
$Y^*_i = \Act^{\Dc}(i,\: Y_{-i}, \: X, \: U)$  and it holds that:
\begin{align*}
			\eex{i \la [n]}{\pr{Y^*_i = Y_i }} > e^{\eps} \cdot \eex{i \la [n]}{\pr{Y^*_i = -Y_i }} + \delta.
\end{align*}
Thus, we get a contradiction since by \cref{prop:hard-to-guess-comp}, algorithm $\Act^{\Dc}$ breaks the \CDP property of $C$.  
\end{proof}

\begin{proof}[Proof of \cref{thm:CDPIP-to-OT}]
Recall that $\Psi$ is an $(\eps,\delta)$-$\DP$ protocol that is $(\ell,0.999)$-acccurate for the inner-product functionality. By \cref{claim:Comp DP to AWEC} there exists a \ppt protocol $\Gamma$ such that $\Gamma\coloneqq\Gamma^\Psi$ is an $(\ell, \alpha=0.001, p=0.001, q=0.001)$-\CompAWEC protocol. By \cref{cor:CompAWEC to CompWEC}, there exists a \ppt protocol $\Lambda$ such that $\Lambda\coloneqq \Lambda^\Gamma$ is $(\epsilon'=\epsilon+0.001,\: p' = p ,\:  q' = 1/2 + 2(q+0.01))$-\CompWEC. Finally, since $44(\eps'+p')<1-q'$ by  \cref{thm:WEC TO OT Comp} there exists a \ppt protocol $\Pi$, such that $\Pi^{\Lambda}$ is a semi-honest computationally secure \OT, as required.
\end{proof}

\fi

\newcommand{\GenRand}{\MathAlgX{GenRand}}
\newcommand{\GenView}{\MathAlgX{GenView}}

\section{\AWEC From \DP Inner Product (Proof of \cref{lemma:DPIP-to-AWEC})}\label{sec:DPIP_to_WAEC}

In this section, we show how to implement AWEC (\cref{def:AWEC}) from an $(\eps,\delta)$-DP channel that is accurate enough for the inner product functionality. We do that using a \ppt protocol and a constructive security proof.



The following protocol is used to prove \cref{lemma:DPIP-to-AWEC}.

\begin{protocol}[Protocol $\Pi = (\Ac,\Bc)$]\label{protocol:DPIP-to-AWEC}
	\item Oracle access: An $(\eps,\delta)$-DP channel $C =((X,U),(Y,V))$ with $X,Y \la \oo^n$\remove{ (i.e., \emph{uniform} channel)} that is $(\ell,0.999)$-accurate for the \emph{inner-product} functionality.
	\item Operation:
	\begin{enumerate}
		
		\item Sample $(x,u), (y,v) \la C$. $\Ac$ gets $(x,u)$ and $\Bc$ gets $(y,v)$. 
		
		\item  $\Ac$ samples  $r \la \zo^n$ and sends $(r,x_{r} = \set{x_i \colon r_i =1})$ to $\Bc$.
		
		\item $\Bc$ samples a random bit $b \la \zo$ and acts as follows:
		
		\begin{enumerate}
			
			\item If $b=0$, it sends $y_{-r}= \set{y_i \colon r_i =0}$ to $\Ac$ and outputs $o_B = \out(v) - \ip{x_r, y_r}$.
			
			\item~\label{step: add noise} Otherwise ($b=1$), it performs the following steps:\label{B_steps_in_abort}
			
			\begin{itemize}
				\item Sample $k$ uniformly random indices $i_1,\ldots,i_k \la [n]$, where $k = \floor{e^{\lambda_1 \eps} \cdot \lambda_2 \cdot \ell^2}$ for constants $\lambda_1,\lambda_2>0$ (to be determined later by the analysis in \cref{eq:lamdas}).
				\item Compute $\ty = (\ty_1,\ldots,\ty_n)$ where $\ty_i \la \begin{cases} \cU_{\oo} & i \in \set{i_1,\ldots,i_k} \\ y_i & \text{otherwise} \end{cases}$,
				\item Send $\ty_{-r}= \set{\ty_i \colon r_i =0}$ to $\Ac$, and output $o_B = \perp$.
			\end{itemize}

		\end{enumerate}
		
		\item  Denote by $\hy_{-r}$ the message $\Ac$ received from $\Bc$. Then $\Ac$ outputs $o_A =  \ip{x_{-r}, \hy_{-r}}$.

	\end{enumerate}
\end{protocol}

In the following, let $\Pi$ be \cref{protocol:DPIP-to-AWEC},
let $C =((X,U),(Y,V))$ be a uniform channel that is $(\ell,0.999)$-accurate for the \emph{inner-product} functionality, let $\tC = ((O_A,V_A),(O_B,V_B))$ be the channel that is induced by $\Pi^C$ (i.e., the parties' outputs and views in the execution of $\Pi$ with oracle access to $C$), and let $R$, $\tY$, $\hY$, $I_1,\ldots,I_k$ be the random variables of the values of $r, \ty, \hy, i_1,\ldots,i_k$ in the execution of $\Pi^C$ (recall that the view of $\Bc$ in the execution is $V_B = (Y, V, R, X_R, I_1,\ldots,I_k, \tY)$, and the view of $\Ac$ is $V_A = (X, U, R, \hY_{-R})$).

Recall that to prove \cref{lemma:DPIP-to-AWEC}, we need to prove that the channel $\tC$ satisfies the accuracy and secrecy properties of \AWEC (\cref{def:AWEC}), and in addition, the secrecy guarantees are constructive as stated in Properties \ref{item:privacy-of-Y}-\ref{item:privacy-of-X} of Lemma~\ref{lemma:DPIP-to-AWEC} (i.e., a violation of at least one secrecy guaranty results with an efficient privacy attack on the DP channel $C$).

We first start with the easy part, which is to prove the accuracy guarantee of $\tC$.

\begin{claim}[Accuracy of $\tC$]
	It holds that
	\begin{align*}
		\pr{\size{O_A - O_B} > \ell \mid O_B \neq \bot} < 0.001. 
	\end{align*}
\end{claim}
\begin{proof}
    Compute
    \begin{align*}
        \pr{\size{O_A - O_B} > \ell \mid O_B \neq \bot}
        &= \pr{\size{\ip{X_{-R},Y_{-R}} - \paren{\out(V) - \ip{X_R,Y_R}}} > \ell}\\
        &= \pr{\size{\out(V) - \ip{X,Y}} > \ell}\\
        &< 0.001,
    \end{align*}
    as required. The first equality holds since conditioned on $O_B \neq \bot$ it holds that $O_A = \ip{X_{-R},Y_{-R}}$, and the inequality holds since $C$ is an $(\ell,0.999)$-accurate for the inner-product functionality. 
\end{proof}

We next move to prove the secrecy guarantees of $\tC$ in a constructive manner. 
In \cref{sec:proving-prop1} we prove property~\ref{item:privacy-of-Y}, and in \cref{sec:proving-prop2} we prove property~\ref{item:privacy-of-X}.


\remove{
	Finally, the following claim (proven in \cref{sec:property 2}) captures the second secrecy guarantee where party $\Bc$ cannot estimate $O_A$ too well, as otherwise, such an estimator $\Bc$ can be used to construct an efficient attack $\Bct$ that violates the privacy guarantee of the DP channel $C$.

	\begin{claim}[Property 2 of Lemma~\ref{lemma:DPIP-to-AWEC}]~\label{clm:property 2}
		There exists an oracle-aided \ppt algorithm $\Bct$ such that for every algorithm $\Bc$ violating the AWEC secrecy property~\ref{AWEC:item:B} of $\tilde{C}$, i.e.
		\begin{align*}
			\pr{\size{\Bc(V_B) - O_A} \leq 1000\ell  \mid O_B=\bot} > q,
		\end{align*}
		where $\frac{131072\ell_2^2}{q^2k}\leq\frac{1}{e^{\varepsilon}+1}$,
		the for $X_i^* = \Bct^{\Bc}(i,\: X_{-i}, \: Y, \: V)$  it holds that
		\begin{align*}
			\eex{i \la [n]}{\pr{X_i^* = X_i \mid X_i^* \neq \bot}} > \frac{e^{\eps}}{e^{\eps}+1}.
		\end{align*}
	\end{claim}
	
	\Enote{Remove:}

	\begin{proof}
		We remind that $\tilde{C}$ is an AWEC channel constructed by a differentially private inner-product protocol. Let $\mathcal{R}=\{i|r_i=0\}$ and $\cI = \set{i_1,\dots,i_k} \cap \cR$ \remove{$\mathcal{I}=\{i|i\in\{i_1,\dots,i_k\},i\in \mathcal{R}\}$}. Then $O_A=o-\ip{x_{\mathcal{R}\setminus \mathcal{I}},\hy_{\mathcal{R}\setminus \mathcal{I}}}-\ip{x_{ \mathcal{I}},\hy_{ \mathcal{I}}}$. It suffices to prove party $\Bc$ cannot predict $\ip{x_{ \mathcal{I}},\hy_{ \mathcal{I}}}$ with high accuracy. We put the details in Section~\ref{sec:property 2}  
	\end{proof}
}

\subsection{B's Security: Proving Property~\ref{item:privacy-of-Y} of \cref{lemma:DPIP-to-AWEC}}\label{sec:proving-prop1}

Let $\Ac$ be an algorithm that violates the AWEC secrecy property~\ref{AWEC:item:A} of $\tilde{C}$ (the channel of $\Pi^C$). Namely,

\begin{align*}
	\size{\pr{\Ac(V_A) = 1 \mid O_B \neq \perp} - \pr{\Ac(V_A) = 1 \mid O_B = \perp}} > \frac1{1000}.
\end{align*}

Recall that $V_A = (X,U,R,\hY_{-R})$ where $\hY = \begin{cases} Y & O_B \neq \perp \\ \tY & O_B = \perp\end{cases}$ and that $\tY$ is obtained from $Y$ by planting uniformly random $\oo$ values in the random locations $I_1,\ldots, I_k$ of $Y$ (the value of $k$ will be determined later by the analysis).

Therefore, the above inequality is equivalent to 

\begin{align*}
	\size{\pr{\Ac(X,U,R,Y_{-R}) = 1} - \pr{\Ac(X,U,R,\tY_{-R}) = 1}} > \frac1{1000}.
\end{align*}

We assume \wlg that $\Ac$ outputs a $\zo$ bit, so the above inequality can be written as

\begin{align}\label{eq:property1:assump}
	\size{\ex{\Ac(X,U,R,Y_{-R}) - \Ac(X,U,R,\tY_{-R})}} > \frac1{1000}.
\end{align}

In the following, for $j \in \set{0,\ldots,k}$, let $\tY^{j} = (\tY^{j}_1,\ldots, \tY^{j}_n)$ where $\tY^{j}_i = \begin{cases} \tY_i & i \in \set{I_1,\ldots,I_j} \\ Y_i & i \notin \set{I_1,\ldots,I_j}\end{cases}$.
Note that $\tY^{0} = Y$ and $\tY^{k} = \tY$. Therefore,

\begin{align}\label{eq:property1:hybrid}
    \lefteqn{\size{\eex{j \la [k]}{\ex{\Ac(X,U,R,Y^{(j-1)}_{-R}) - \Ac(X,U,R,\tY^{(j)}_{-R})}}}}\\
    &= \frac1k \cdot \size{\sum_{j=1}^k \ex{\Ac(X,U,R,Y^{(j-1)}_{-R})} - \ex{\Ac(X,U,R,Y^{(j)}_{-R})}}\nonumber\\
    &= \frac1k \cdot \size{\ex{\Ac(X,U,R,Y_{-R}) - \Ac(X,U,R,\tY_{-R})}}\nonumber\\
    &> \frac1{1000 k},\nonumber
\end{align}
where the inequality holds by \cref{eq:property1:assump}.

In the following, let $J \la [n]$ (sampled independently of the other random variables defined above), let $Z = (Z_1,\ldots,Z_n) = \tY^{J-1}$ and let $Z^{(i)} = (Z_1,\ldots,Z_{i-1}, -Z_i, Z_{i+1},\ldots, Z_n)$. It holds that

\begin{align}\label{eq:property1:exp}
	\lefteqn{\size{\ex{\Ac(X,U,R,Z_{-R})} - \ex{\Ac(X,U,R,Z^{(I_{J})}_{-R})}}}\\
        &= \frac12 \cdot \size{\ex{\Ac(X,U,R,Y^{(J-1)}_{-R}) - \Ac(X,U,R,\tY^{(J)}_{-R})}}\nonumber\\
        &= \frac12\cdot \size{\eex{j \la [k]}{\ex{\Ac(X,U,R,Y^{(j-1)}_{-R}) - \Ac(X,U,R,\tY^{(j)}_{-R})}}}\nonumber\\
        &> \frac1{2000 k},\nonumber
\end{align}
where the first equality holds since $\tY^{J}$ is equal to $Z$ w.p. $1/2$ (happens when $\tY_{I_{J}} = Y_{I_{J}}$)  and otherwise is equal to $Z^{(I_{J})}$ (happens when $\tY_{I_{J}} \neq Y_{I_{J}}$), and the last inequality holds by \cref{eq:property1:hybrid}.


\iffull
We next use the following lemma, proven in \cref{sec:prediction-lemma} (recall that a proof overview appears in \cref{sec:overview:prediction-lemma}).
\fi

\def\PredictionLemma{
    For every $\gamma \in (0,1)$ and $n \in \bbN$, there exists an oracle aided (randomized) algorithm $\Gc = \Gc_{\gamma} \colon [n] \times \zo^{n-1} \rightarrow \set{-1,1,\perp}$ that runs in time $\poly(n,1/\gamma)$ such that the following holds: 
	
	Let $(Z,W) \in \oo^n \times \zo^*$ be jointly distributed random variables, let $R \la \zo^n$ and $I \la [n]$ (sampled independently), 
	and let $F$ be a (randomized) function that satisfies
	\begin{align*}
		\size{\ex{F(R,Z_{R},W) - F(R,Z^{(I)}_{R},W)}} \geq \gamma,
	\end{align*}
	for $Z^{(I)} = (Z_1,\ldots, Z_{I-1}, -Z_I, Z_{I+1},\ldots, Z_n)$. Then
	\begin{enumerate}
		\item $\pr{\Gc^F(I,Z_{-I}, W) = -Z_I} \leq O\paren{\frac{1}{\gamma^2 n}}$, and\label{property1:prediction:bad}
		
		\item $\pr{\Gc^F(I,Z_{-I}, W) = Z_I} \geq \Omega(\gamma) - O\paren{\frac{1}{\gamma^2 n}}$.\label{property1:prediction:good}
	\end{enumerate}
}

\begin{lemma}\label{lemma:property1:prediction}
    \PredictionLemma
\end{lemma}

We now ready to finalize the proof of Property~\ref{item:privacy-of-Y} of \cref{lemma:DPIP-to-AWEC} using \cref{eq:property1:exp,lemma:property1:prediction}.

\begin{proof}
	
	Consider the following oracle-aided algorithm $\Act$:
	
	\begin{algorithm}[Algorithm $\Act$]\label{alg:Act}
		\item Inputs: $i \in [n]$, $y_{-i} \in \oo^{n-1}$, $x \in \oo^n$ and $u \in \zo^*$.
		\item Oracle: Deterministic algorithm $\Ac$.
		\item Operation:~
		\begin{enumerate}
			\item Sample $j \la [k]$ and $i_1,\ldots,i_{j-1} \la [n]$.
			
			\item Sample $z_{-i} = (z_1,\ldots,z_{i-1}, z_{i+1},\ldots, z_n)$, where for $t \in [n]\setminus \set{i}$: $$z_{t} \la \begin{cases} \cU_{\oo} & t \in \set{i_1,\ldots,i_{j-1}} \\ y_t & \text{otherwise} \end{cases}.$$\label{step:z-i}
			
			\item Output $y_i^* \la \Gc^{F}(i,z_{-i},w)$ for $w=(x,u)$, where $\Gc = \Gc_{\frac1{2000 k}}$ is the algorithm from \cref{lemma:property1:prediction}, and $F(r,z_r,w) \eqdef \Ac(w,(1-r_1,\ldots,1-r_n), z_r)$.
		\end{enumerate}
	\end{algorithm}

	Note that by \cref{eq:property1:exp} it holds that
	\begin{align*}
		\size{\ex{F(R,Z_{R}, W)} - \ex{F(R,Z_{R}^{I},W)}} > \frac1{2000 k}
	\end{align*}
	for $W = (X,U)$, $I = I_J$ and the function $F(r,z_r,w=(x,u)) \eqdef \Ac(w,(1-r_1,\ldots,1-r_n), z_r)$. Thus, \cref{lemma:property1:prediction} implies that
	\begin{align}\label{eq:G:LB}
		\pr{\Gc^{F}(I,Z_{-I}, V) = -Z_I} \leq O\paren{\frac{k^2}{n}}
	\end{align}
	and
	\begin{align}\label{eq:G:UB}
		\pr{\Gc^{F}(I,Z_{-I}, V) = Z_I} \geq \Omega(1/k) - O\paren{\frac{k^2}{n}}.
	\end{align} 
	
	Recall that we denote $Y_i^* = \Act^{\Ac}(i,Y_{-i},X,U)$. We next lower-bound $\eex{i \la [n]}{\pr{Y_i^* = Y_i}}$ and upper-bound $\eex{i \la [n]}{\pr{Y_i^* = -Y_i}}$.
	The first bound hold by the following calculation
	\begin{align}\label{eq:Y_i*-LB}
		\eex{i \la [n]}{\pr{Y_i^* = Y_i}}
		&\geq 0.9 \cdot \eex{i \la [n]}{\pr{Y_i^* = Z_i}}\\
		&= 0.9 \cdot \eex{i \la [n]}{\pr{\Gc^{F}(i,Z_{-i},W) = Z_i}}\nonumber\\
		&\geq \Omega\paren{\frac1k} - O\paren{\frac{k^2}{n}}.\nonumber
	\end{align}
	The first inequality holds since $i \notin \set{I_1,\ldots,I_{J-1}} \implies Z_i = Y_i$ and $\pr{i \notin \set{I_1,\ldots,I_{J-1}}} \geq \paren{1 - \frac1n}^{k-1} \geq 0.9$ (recall that $k \in o(n)$). The equality holds since, conditioned on $X=x,U=u,Y_{-i} = y_{-i}$ (the inputs of $\Act$), the value of $z_{-i}$ that is sampled in \stepref{step:z-i} of $\Act$ is distributed the same as $Z_{-i}|_{X=x,U=u,Y_{-i} = y_{-i}}$, and therefore, $Y_i^* \equiv \Gc^{F}(i,Z_{-i},W)$ for $W = (X,U)$. The last inequality holds by \cref{eq:G:UB}.
	
	On the other hand, we have that
	\begin{align}\label{eq:Y_i*-UB}
		\eex{i \la [n]}{\pr{Y_i^* = -Y_i}}
		&\leq \eex{i \la [n]}{\pr{Y_i^* = - Y_i \mid Y_i = Z_i} + \pr{Z_i \neq Y_i}}\\
		&\leq \eex{i \la [n]}{\pr{Y_i^* = - Z_i}} + \frac{k}{2n}\nonumber\\
		&=  \eex{i \la [n]}{\pr{\Gc^{F}(i,Z_{-i}) = - Z_i}} + \frac{k}{2n}\nonumber\\
		&\leq O\paren{\frac{k^2}{n}}.\nonumber
	\end{align}
	The second inequality holds since for any fixing $Y_{-i}=y_{-i},X=x,U=u$, we have $Y_i^*= \Act^{\Ac}(i,y_{-i},x,u)$ which is independent of $(Y_i,Z_i)$, and also since $i \notin \set{I_1,\ldots,I_{J-1}} \implies Z_i = Y_i$ which yields that
	\begin{align*}
		\pr{Z_i \neq Y_i} \leq 1 - \pr{i \notin \set{I_1,\ldots,I_{J-1}}} \leq 1 - \paren{1 - \frac1n}^{k} \leq 1 - e^{\frac{k}{2n}} \leq \frac{k}{2n}. 
	\end{align*}
	The equality in \cref{eq:Y_i*-UB} holds since, conditioned on $X=x,U=u,Y_{-i} = y_{-i}$ (the inputs of $\Act$), the value of $z_{-i}$ that is sampled in \stepref{step:z-i} of $\Act$ is distributed the same as $Z_{-i}|_{X=x,U=u,Y_{-i} = y_{-i}}$, and therefore, $Y_i^* \equiv \Gc^{F}(i,Z_{-i},W)$ for $W = (X,U)$. The last inequality in \cref{eq:Y_i*-UB} holds by \cref{eq:G:LB}.
	
	By \cref{eq:Y_i*-LB,eq:Y_i*-UB}, assuming that $\delta \leq 1/n$ where $n$ is large enough, 
	there exists a constant $c > 0$ such that if $k \leq c \cdot (e^{-\eps} n)^{1/3}$ then $\eex{i \la [n]}{\pr{Y_i^* = Y_i}} > e^{\eps} \cdot \eex{i \la [n]}{\pr{Y_i^* \neq Y_i}} + \delta$, as required.
	Recall that $k = \ceil{e^{\lambda_1 \eps} \lambda_2 \cdot \ell^2}$, where $\lambda_1,\lambda_2$ are the constants from \cref{eq:lamdas}. Hence, we can set a bound of $e^{-c_1 \eps} c_2 \cdot n^{1/6}$ on $\ell$ for $c_1 = \lambda_1/2 + 1/6$ and $c_2 = \sqrt{c/\lambda_2}$ to guarantee that $k \leq c \cdot (e^{-\eps} n)^{1/3}$.
	
\end{proof}

\remove{
\subsubsection{Proving \cref{lemma:property1:prediction}}\label{sec:prediction-lemma}
\begin{proof}
	
	Let $0 < \gamma \leq 0.01$ and consider the following algorithm $\Gc = \Gc_{\gamma}$:
	\begin{algorithm}
		\item Inputs: $i \in [n]$,  $z_{-i} = (z_1,\ldots,z_{i-1},z_{i+1},\ldots,z_n) \in \oo^{n-1}$ and $w \in \zo^*$.
		\item Parameter: $0 < \gamma \leq 0.01$.
		\item Oracle: $f \colon \zo^n \times \oo^{\leq n} \times \zo^* \rightarrow \zo$.
		\item Operation:~
		\begin{enumerate}
			\item For $b \in \oo$:
			\begin{enumerate}
				\item Let $z^b = (z_1,\ldots,z_{i-1},b,z_{i+1},\ldots,z_n)$.
				\item Estimate $\mu^b \eqdef \eex{r \la \zo^n \mid r_i = 1, \: f}{f(r,z^b_{r}, w)}$ as follows:
				\begin{itemize}
					\item Sample $r_1,\ldots,r_{s} \la \set{r \in \zo \colon r_i = 1}$, for $s = \ceil{\frac{128 \log (6n)}{\gamma^2}}$, and then sample $\tilde{\mu}^b \la \frac1{s} \sum_{j=1}^s f(r_j,z^b_{r_j}, w)$ (using $s$ calls to the oracle $f$).
				\end{itemize}
			\end{enumerate}
			\item Estimate $\mu^* \eqdef \eex{r \la \zo^n \mid r_i = 0, \: f}{f(r,z^1_{r}, w)}$ as follows:
			\begin{itemize}
				\item Sample $r_1,\ldots,r_{s}  \la \set{r \in \zo \colon r_i = 0}$, for $s = \ceil{\frac{128 \log (6n)}{\gamma^2}}$, and then sample $\tilde{\mu}^* \la \frac1{s} \sum_{j=1}^s f(r_j,z^1_{r_j}, w)$ (using $s$ calls to the oracle $f$).
			\end{itemize}
			\item If exists $b \in \oo$ s.t. $\size{\tilde{\mu}^b - \tilde{\mu}^*} < \gamma/4$ and $\size{\tilde{\mu}^{-b} - \tilde{\mu}^*} > \gamma/4$, output $b$.
			\item Otherwise, output $\bot$.
		\end{enumerate}
	\end{algorithm}
	In the following, fix a pair of (jointly distributed) random variables $(Z,W) \in \oo^n \times \zo^*$ and a randomized function  $f \colon \zo^n \times \oo^{\leq n} \times \zo^* \rightarrow \oo$ that satisfy 
	\begin{align*}
		\size{\ex{f(R,Z_{R},W) - f(R,Z^{(I)}_{R},W)}} \geq \gamma,
	\end{align*}
	for $R \la \zo^n$ and $I \la [n]$ that are sampled independently. 
	Our goal is to prove that 
	
	\begin{align}\label{eq:predic-lemma:UB}
		\pr{\Gc^f(I,Z_{-I}, W) = -Z_I} \leq O\paren{\frac{1}{\gamma^2 n}},
	\end{align}
	and 
	\begin{align}\label{eq:predic-lemma:LB}
		\pr{\Gc^f(I,Z_{-I}, W) = Z_I} \geq \Omega(\gamma) - O\paren{\frac{1}{\gamma^2 n}}.
	\end{align}

	Note that
	\begin{align}\label{eq:D}
		\text{For every random variable }D \in [-1,1]\text{ with }\size{\ex{D}} \geq \gamma >0: \quad \pr{\size{D} > \gamma/2} > \gamma/2,
	\end{align}
	as otherwise, $\size{\ex{D}} \leq \ex{\size{D}} \leq 1\cdot \frac{\gamma}{2} + \frac{\gamma}{2}(1-\frac{\gamma}{2}) < \gamma$.

	By applying \cref{eq:D} with $D = \eex{r \la \zo^n, f}{f(r,Z_{r},W) - f(r,Z^{(I)}_{r},W)}$, it holds that
	\begin{align}\label{eq:main-lemma:good}
		\ppr{(i,z,w) \la (I,Z,W)}{\size{\eex{r \la \zo^n, f}{f(r,z_{r},w) - f(r,z^{(i)}_{r},w)}} > \gamma/2} > \gamma/2.
	\end{align}
	On the other hand, for every fixing of $(z,w) \in \Supp(Z,W)$, we can apply \cref{lem:distance-I} with the function $f_{z,w}(r) = f(r,z_{r},w)$ and with $\alpha =\frac{\gamma}{16}$ to obtain that
	\begin{align*}
		\ppr{i\la I}{\: \size{\eex{r \la \zo^n \mid r_i = 0, \: f}{f(r,z_{r},w)} - \eex{r \la \zo^n \mid r_i = 1, \: f}{f(r,z_{r},w)} \:} \geq \frac{\gamma}{16}} \leq \frac{512}{n \gamma^2}.
	\end{align*}
	But since the above holds for every fixing of $(z,w)$, then in particular it holds that
	\begin{align}\label{eq:main-lemma:bad}
		\ppr{(i,z,w) \la (I,Z,W)}{\: \size{\eex{r \la \zo^n \mid r_i = 0, \: f}{f(r,z_{r},w)} - \eex{r \la \zo^n \mid r_i = 1, \: f}{f(r,z_{r},w)} \:} \geq \frac{\gamma}{16}} \leq \frac{512}{n \gamma^2}.
	\end{align}
	
	We next prove \cref{eq:predic-lemma:UB,eq:predic-lemma:LB} using \cref{eq:main-lemma:good,eq:main-lemma:bad}.
	
	In the following, for a triplet $t = (i,z,w) \in \Supp(I,Z,W)$, consider a random execution of $\Gc^f(i,z_{-i},w)$. For $x \in \set{-1,1,*}$, let $\mu^x_t$ be the value of $\mu^x$ in the execution, and let $M^x_t$ be the (random variable of the) value of $\tilde{\mu}^x$ in the execution. Note that by definition, it holds that $\mu_{i,z,w}^{z_i} = \eex{r \la \zo^n, f}{f(r,z_{r},w)}$ and $\mu_{i,z,w}^{-z_i} = \eex{r \la \zo^n, f}{f(r,z^{(i)}_{r},w)}$. Therefore, \cref{eq:main-lemma:good} is equivalent to 
	\begin{align}\label{eq:main-lemma:good2}
		\ppr{t=(i,z,w) \la (I,Z,W)}{\size{\mu_t^{z_i} - \mu_t^{-z_i}} > \gamma/2} > \gamma/2.
	\end{align}
	Furthermore, note that $\mu_{i,z,w}^* \eqdef \eex{r \la \zo^n \mid r_i = 0, \: f}{f(r,z^1_{r}, w)} = \eex{r \la \zo^n \mid r_i = 0, \: f}{f(r,z_{r}, w)}$ and that $\mu_{i,z,w}^{z_i} = \eex{r \la \zo^n \mid r_i = 1, \: f}{f(r,z^{z_i}_{r}, w)}$. Therefore, \cref{eq:main-lemma:bad} is equivalent to 
	\begin{align}\label{eq:main-lemma:bad2}
		\ppr{t = (i,z,w) \la (I,Z,W)}{\: \size{\mu_t^* - \mu_t^{z_i}}\geq \frac{\gamma}{16}}  \leq \frac{512}{n \gamma^2}.
	\end{align}
	
	We next prove the lemma using \cref{eq:main-lemma:good2,eq:main-lemma:bad2}.
	
	Note that by Hoeffding's inequality, for every $t = (i,z,w) \in \Supp(I,Z,W)$ and  $x \in \set{-1,1,*}$ it holds that $\pr{\size{M_t^x - \mu_t^x} \geq \frac{\gamma}{16}} \leq 2\cdot e^{-2 s \paren{\frac{\gamma}{16}}^2} \leq \frac1{3n}$, which yields that for every fixing of $t = (i,z,w) \in \Supp(I,Z,W)$, w.p.\ at least $1-1/n$ we have for all $x \in \set{-1,1,*}$ that $\size{M_t^x - \mu_t^x} < \frac{\gamma}{16}$ (denote this event by $E_t$).
	
	The proof of \cref{eq:predic-lemma:UB} holds by the following calculation:
	\begin{align*}
		\lefteqn{\pr{\Gc^f(I,Z_{-I},W) = -Z_I}}\\
		&= \eex{(i,z,w) \la (I,Z,W)}{\pr{\Gc^f(i,z_{-i},w) = -z_i}}\\
		&= \eex{t = (i,z,w) \la (I,Z,W)}{\pr{\set{\size{M_t^* - M_t^{z_i}} > \gamma/4} \land \set{\size{M_t^* - M_t^{-z_i}} < \gamma/4}}}\\
		&\leq \eex{t =(i,z,w) \la (I,Z,W)}{\pr{\size{M_t^* - M_t^{z_i}} > \gamma/4}}\\
		&\leq \eex{t = (i,z,w) \la (I,Z,W)}{\pr{\size{M_t^* - M_t^{z_i}} > \gamma/4 \mid E_t}} + \frac{1}{n}\\
		&\leq \ppr{t = (i,z,w) \la (I,Z,W)}{\size{\mu_t^* -  \mu_t^{z_i}} \geq \frac{\gamma}{16}} + \frac{1}{n}\\
		&\leq \frac{512}{n \gamma^2} + \frac1n,
	\end{align*}
	The second inequality holds since $\pr{\neg E_t} \leq 1/n$ for every $t$. The penultimate inequality holds since conditioned on $E_t$, it holds that $\size{M_t^* - \mu_t^*} \leq \frac{\gamma}{16}$ and $\size{M_t^{z_i} - \mu_t^{z_i}} \leq \frac{\gamma}{16}$, which implies that $\size{M_t^* - M_t^{z_i}} > \gamma/4 \: \implies \: \size{\mu_t^* -  \mu_t^{z_i}} \geq \frac{\gamma}{4} - 2\cdot \frac{\gamma}{16} > \frac{\gamma}{16}$. The last inequality holds by \cref{eq:main-lemma:bad2}.
	
	It is left to prove \cref{eq:predic-lemma:LB}. 
	Observe that \cref{eq:main-lemma:good2,eq:main-lemma:bad2} imply with probability at least $\gamma/2 - \frac{512}{n \gamma^2}$ over $t = (i,z,w)\in (I,Z,W)$, we have $\size{\mu_t^* -  \mu_t^{z_i}} \leq  \frac{\gamma}{16}$ and $\size{\mu_t^{z_i}- \mu_t^{-z_i}} \geq\frac{\gamma}{2}$. Therefore, conditioned on event $E_t$ (i.e., $\forall x \in \set{-1,1,*}: \: \size{M_t^{x} - \mu_t^{x}} \leq \frac{\gamma}{16}$), we have for such triplets $t = (i,z,w)$ that 
	\begin{align*}
		\size{M_t^* - M_t^{-z_i}} 
		& \geq \size{\mu_t^{z_i} - \mu_t^{-z_i}} - \size{\mu_t^{z_i} - \mu_t^*} - \size{M_t^*  - \mu_t^*} -  \size{M_t^{-z_i} - \mu_t^{-z_i}}\\
		&\geq \frac{\gamma}{2} - 3\cdot \frac{\gamma}{16}\\
		&> \gamma/4,
	\end{align*}
	and 
	\begin{align*}
		\size{M_t^* - M_t^{z_i}}
		&\leq  \size{M_t^* - \mu_t^*} + \size{\mu_t^* - \mu_t^{z_i}} + \size{\mu_t^{z_i} - M_t^{z_i}} \\
		&\leq 3\cdot \frac{\gamma}{16}\\
		&< \gamma/4.
	\end{align*}
	
	Thus, we conclude that
	\begin{align*}
		\lefteqn{\pr{\Gc^f(I,Z_{-I},W) = Z_I}}\\
		&= \eex{(i,z,w) \la (I,Z,W)}{\pr{\Gc^f(i,z_{-i},w) = z_i}}\\
		&= \eex{t = (i,z,w) \la (I,Z,W)}{\pr{\set{\size{M_t^* - M_t^{-z_i}} > \gamma/4} \land \set{\size{M_t^* - M_t^{z_i}} < \gamma/4}}}\\
		&\geq \paren{1 - \frac1{n}}\cdot \eex{t = (i,z,w) \la (I,Z,W)}{\pr{\set{\size{M_t^* - M_t^{-z_i}} > \gamma/4} \land \set{\size{M_t^* - M_t^{z_i}} < \gamma/4} \mid E_t}}\\
		&\geq \paren{1 - \frac1{n}}\cdot \paren{\gamma/2 - \frac{512}{n \gamma^2}}\\
		&\geq \gamma/4 - \frac{1024}{n \gamma^2},
	\end{align*}
	which proves \cref{eq:predic-lemma:LB}. The first inequality holds since $\pr{E_t} \geq 1-1/n$ for every $t = (i,z,w)$, and the second one holds by the observation above.
	
\end{proof}
}

\subsection{A's Security: Proving Property~\ref{item:privacy-of-X} of \cref{lemma:DPIP-to-AWEC}}\label{sec:proving-prop2}

Let $\Bc$ be an algorithm that violates the AWEC secrecy property~\ref{AWEC:item:B} of $\tilde{C} = ((V_A,O_A),(V_B,O_B))$ --- the channel of $\Pi^{C = ((X,U),(Y,V))}$ (\cref{protocol:DPIP-to-AWEC}). Namely,

\begin{align}\label{eq:violating-B}
	\pr{\size{\Bc(V_B) - O_A} \leq 1000 \ell \mid O_B=\bot} > \frac1{1000}.
\end{align}

Recall that $V_B = (Y,V,R,X_{R},I_1,\ldots,I_k, \tY)$ where $I_1,\ldots,I_k \la [n]$ are the indices that $\Bc$ chooses at Step~\ref{B_steps_in_abort}, and $\tY = (\tY_1,\ldots,\tY_n)$ where $\tY_i \la \oo$ for $i \in \set{I_1,\ldots,I_k}$ and otherwise $\tY_i = Y_i$. Furthermore, conditioned on $O_B=\bot$, recall that $O_A = \ip{X_{-R}, \tY_{-R}}$. Therefore, \cref{eq:violating-B} is equivalent to 
\begin{align}\label{eq:Bc-guarantee}
	\pr{\size{\Bc(V_B) - \ip{X_{-R}, \tY_{-R}}} \leq 1000 \ell } > \frac1{1000}.
\end{align}

In the following, define the random variable $H$ to be the first $m = \ceil{k/4}$ indices of $R^0 \cap \set{I_1,\ldots,I_k}$ for $R^0 = \set{i \colon R_i = 0}$, where we let $H = \emptyset$ if the size of the intersection is smaller than $m$.
Since $R \la \zo^n$, Hoeffding's inequality implies that $\pr{\size{R^0} \geq 0.4 n} \geq 1 - e^{-\Omega(n)}$. Since $I_1,\ldots,I_k \la [n]$ (independent of $R$), then again by Hoeffding's inequality we obtain that $\pr{\size{H} = \ceil{k/4} \mid \size{R^0} \geq 0.4 n} \geq 1 - e^{-\Omega(k)}$, which yields that
\begin{align*}
	\pr{H \neq \emptyset} = \pr{\size{H} = \ceil{k/4}} \geq 1 - e^{-\Omega(k)} - e^{-\Omega(n)} \geq 1-\frac1{10000}.
\end{align*}
Therefore, by the union bound,
\begin{align}\label{eq:good-and-not-empty-H}
	\lefteqn{\pr{\set{\size{\Bc(V_B) - \ip{X_{-R}, \tY_{-R}}} \leq 1000 \ell} \land \set{H \neq \emptyset}}}\\
	&= 1- \pr{\set{\size{\Bc(V_B) - \ip{X_{-R}, \tY_{-R}}} > 1000 \ell} \lor \set{H = \emptyset}}\nonumber\\
	&\geq 1- \pr{\set{\size{\Bc(V_B) - \ip{X_{-R}, \tY_{-R}}} > 1000 \ell}} - \pr{ \set{H = \emptyset}}\nonumber\\
    &\geq 1- \paren{1- \frac1{1000}} - \frac1{10000} \nonumber\\
     &\geq \frac1{2000}.\nonumber
\end{align}

In the following, let $d$ be the number of random coins that $\Bc$ uses, and for $\psi \in \zo^d$ let $\Bc_{\psi}$ be algorithm $\Bc$ when fixing its random coins to $\psi$.
Let $\Psi \la \zo^d$,  $Z = X_H$,  $T' = (Y,V,R,I_1,\ldots,I_k, H,\tY_{-\cH}, X_{-\cH})$,  $T= (\Psi, T')$ and $S = \tY_H$. Note  that conditioned on $H \neq \perp$, $S$ is a uniformly random string in $\oo^m$, independent of $Z$ and $T$, and note that $V_B$ is a deterministic function of $(\tY_{H}, T)$ (because $X_R$, which is part of $V_B$, is a sub-string of $X_{-\cH}$).
\cref{eq:good-and-not-empty-H} yields that w.p. $1/2000$ over $z \la Z$,  $t = (\psi, t'=(y,v,r,i_1,\ldots,i_k,\cH,x_{-\cH},\ty_{-\cH})) \la T$ and $s \la \oo^m$,
the following holds for $\bar{\cH} =  \set{i \in [n] \colon r_i = 0} \setminus \cH\:$:
\begin{align*}
	\size{\Bc_{\psi}(s,t') - \ip{x_{\bar{\cH}}, \ty_{\bar{\cH}}} - \ip{z, s}}\leq 1000 \ell.
\end{align*}

By denoting $f(s,t=(\psi,t')) = \Bc_{\psi}(s,t') + \ip{x_{\bar{\cH}}, \ty_{\bar{\cH}}}$,\footnote{Note that $f(s,t)$ is well-defined because $t$ contains $x_{\bar{\cH}}$ and $\ty_{\bar{\cH}}$ (sub-strings of $x_{-\cH}$ and $\ty_{-\cH}$, respectively).} the above observation is equivalent to

\begin{align}\label{eq:our-good-f}
	\ppr{(z,t) \la (Z,T), \: s \la \oo^m}{\size{f(s,t) - \ip{z, s}} \leq 1000 \ell} \geq \frac1{2000}.
\end{align}

In the following, let $\cD$ be the joint distribution of $(Z,T)$, which is equivalent to the output of $\GenView^{C}()$ defined below in \cref{alg:GenView}.

\begin{algorithm}[$\GenRand$]\label{alg:GenRand}
	~
	\begin{enumerate}
            \item Sample $\psi \la \zo^d$.
		\item Sample $r \la \zo^n$ and $i_1,\ldots,i_k \la [n]$.
		\item Let $\cH$ be the first $m=\ceil{k/4}$ indices of $\set{i \in [n] \colon r_i = 0} \cap  \set{i_1,\ldots,i_k}$, where $\cH = \emptyset$ if the intersection size is less than $m$.
		\item Let $\bar{\cH} = \set{i \in [n] \colon r_i = 0} \setminus \cH$.
		\item Sample $\ty_{\bar{\cH}} \la \oo^{\size{\bar{\cH}}}$.
		\item Output $(\psi,r,i_1,\ldots,i_k,\cH,\ty_{\bar{\cH}})$.
	\end{enumerate}
\end{algorithm}

\begin{algorithm}[$\GenView$]\label{alg:GenView}
	\item Oracle: A channel $C = ((X,U),(Y,V))$.
	\item Operation:~
	\begin{enumerate}
		\item Sample $((x,u),(y,v)) \la C$.
		\item Sample $(\psi,r,i_1,\ldots,i_k,\cH,\ty_{\bar{\cH}}) \la \GenRand()$ (\cref{alg:GenRand}).
		\item Output $z = x_{\cH}$ and $t = (\psi, t')$ for $t'=(y,v,r,i_1,\ldots,i_k,\cH,\ty_{-\cH},x_{-\cH})$, where $\ty_i = y_i$ for $i \in [n]\setminus (\cH \cup \bar{\cH})$.
	\end{enumerate}
\end{algorithm}

We now can use the following reconstruction result from \cite{HaitnerMST22}:

\begin{fact}[Follows by Theorem 4.6 in \cite{HaitnerMST22}]\label{fact:prev-rec}
	There exists constants $\eta_1,\eta_2 > 0$ and a \ppt algorithm $\Dist$ such that the following holds for large enough $m \in \bbN$: Let $\eps \geq 0$ and $a \geq \log m$, and let $\cD$ be a distribution over $\oo^m \times \Sigma^*$. Then for every function $f$ that satisfies 
	\begin{align*}
		\ppr{(z,t) \la \cD, \: s \la \oo^m}{f(s,t) - \ip{z,s} \leq a} \geq e^{\eta_1 \cdot \eps}\cdot \eta_2 \cdot a/\sqrt{m},
	\end{align*}
	it holds that
	\begin{align*}
		\ppr{(z,t) \la \cD, \: j \la [m]}{\Dist^{\cD,f}(j, z, t) = 1} > e^{\eps}\cdot \ppr{(z,t) \la \cD, \: j \la [m]}{\Dist^{\cD,f}(j ,z^{(j)}, t) = 1} + \frac1m,
	\end{align*}
	where $z^{(j)} = (z_1,\ldots,z_{j-1},-z_j,z_{j+1},\ldots,z_m)$.\footnote{Theorem 4.6 in \cite{HaitnerMST22} actually considered a harder setting where $z$ is the coordinate-wise product of two vectors $x,y \in \oo^n$, and $f$ only guarantees accuracy when in addition to $s$ and $t$, it also gets $x_{s} = \set{x_i \colon s_i = 1}$ and $y_{-s} = \set{y_i \colon s_i = -1}$ as inputs.}
\end{fact}

Now, we would like to apply \cref{fact:prev-rec} with $\cD = \GenView^C()$ and $a = 1000 \ell$. To do that, \cref{fact:prev-rec} yields that we need to choose $k$ such that $\frac{e^{\eta_1 \cdot \eps}\cdot \eta_2  \cdot 1000\ell}{\ceil{k/4}} \leq \frac1{2000}$, which holds by choosing $k = \floor{e^{\lambda_1 \eps}\cdot \lambda_2 \cdot \ell^2}$ with 
\begin{align}\label{eq:lamdas}
	\lambda_1 = \eta_1 \text{ and }\lambda_2 = 10^7 \eta_2 + 1,
\end{align}
where $\eta_1,\eta_2$ are the constants from \cref{fact:prev-rec}.

We deduce from \cref{fact:prev-rec,eq:our-good-f} that 

\begin{align}\label{eq:Dist-gap}
	\ppr{(z,t) \la \cD, \: j \la [m]}{\Dist^{\cD,f}(j, z, t) = 1} > e^{\eps}\cdot \ppr{(z,t) \la \cD, \: j \la [m]}{\Dist^{\cD,f}(j ,z^{(j)}, t) = 1} + \frac1m.
\end{align}

We now ready to describe our algorithm $\Bct$ that satisfies Property~\ref{item:privacy-of-X} of \cref{lemma:DPIP-to-AWEC}.

\begin{algorithm}[Algorithm $\Bct$]\label{alg:Bct}
	\item Inputs: $i \in [n]$, $x_{-i} = (x_1,\ldots,x_{i-1},x_{i+1},\ldots,x_n) \in \oo^{n-1}$, $y \in \oo^n$ and $v \in \zo^*$.
	\item Oracles: A channel $C = ((X,U),(Y,V))$ and an algorithm $\Bc$.
	\item Operation:~
	\begin{enumerate}
		\item Sample $(\psi, r,i_1,\ldots,i_k,\cH,\ty_{\bar{\cH}}) \la \GenRand()$ (\cref{alg:GenRand}).
		\item If $i \notin \cH$, output $\bot$.
		\item Otherwise:
		\begin{enumerate}
			\item Let $t = (y,v,r,i_1,\ldots,i_k,\cH,\ty_{-\cH},x_{-\cH})$ where $\ty_{i'} = y_{i'}$ for $i' \in [n]\setminus (\cH \cup \bar{\cH})$.
			\item For $b \in \oo$: Let $x^b = (x_1,\ldots,x_{i-1},b, x_{i+1},\ldots,x_n)$ and $z^b = x^b_{\cH} \in \oo^m$, where $m = \ceil{k/4}$.
			\item Let $j \in [m]$ be the index such that $z^1_j \neq z^{-1}_j$.
			\item Sample $b \la \oo$ and $o \la \Dist^{\GenView^C, f}(j,z^b,t)$ where $\GenView^C$ is \cref{alg:GenView} with oracle access to $C$, and $f(s,t = (\psi,t')) = \Bc_{\psi}(s,t') + \ip{x_{\bar{\cH}}, \ty_{\bar{\cH}}}$. 
			\color{gray}{\# Recall that $x_{\bar{\cH}}$ is a sub-string of $x_{-\cH}$ (part of $t$) and that $ \ty_{\bar{\cH}}$ is a sub-string of $\ty_{-\cH}$ (also part of $t$).}
			\color{black}{\item If $o = 1$, output $b$. Otherwise, output $\bot$.}
		\end{enumerate}
	\end{enumerate}
\end{algorithm}

\begin{proof}[Proof of Property~\ref{item:privacy-of-X} of \cref{lemma:DPIP-to-AWEC} using \cref{alg:Bct}]
	
In the following, let $\cD = \GenView^C()$ and $((X,U),(Y,V)) \la C$. Consider a random execution of $\Bct^{\Bc, C}(I,\: X_{-I}, \: Y, \: V)$ for $I \la [n]$, and let $T, H, B, O, Z, J$ be the values of $\:\: t, h,\cH,o, z, j$ in the execution. 
Let $p = \ppr{(z,t) \la \cD, \: j \la [m]}{\Dist^{\cD,f}(j ,z^{(j)}, t) = 1}$, and note that conditioned on $I \in H$, $J$ is distributed uniformly over $[m]$. Therefore,
$$\pr{\Dist^{\cD,f}(J, Z^{(J)}, T) = 1 \mid I \in H} = p$$ and $$\pr{\Dist^{\cD,f}(J, Z, T) = 1 \mid I \in H} = \ppr{(z,t) \la \cD, \: j \la [m]}{\Dist^{\cD,f}(j ,z, t) = 1} \geq e^\eps p + \frac1m,$$ where the inequality holds by \cref{eq:Dist-gap}.
Thus, the following holds for $X^*_i = \Bct^{\Bc, C}(i,\: X_{-i}, \: Y, \: V)$:

\begin{align*}
	\eex{i \la [n]}{\pr{X^*_i = -X_i}}
	&= \pr{X_I^* = -X_I}\\
	&=\pr{\set{I \in H} \land \set{B = -X_I} \land \set{O=1}}\\
	&= \paren{\frac{m}{n}\cdot \pr{H \neq \emptyset}} \cdot \frac12 \cdot  \pr{O=1 \mid \set{I \in H}\land \set{B = -X_I}}\\
	&= \frac{m}{2n} \cdot \pr{H \neq \emptyset} \cdot \pr{\Dist^{\cD,f}(J, Z^{(J)}, T) = 1 \mid I \in H}\\
	&= \frac{m}{2n} \cdot \pr{H \neq \emptyset} \cdot p,
\end{align*}
 and 

\begin{align*}
	\eex{i \la [n]}{\pr{X^*_i = X_i}}
	&=\pr{\set{I \in H} \land \set{B = X_I} \land \set{O=1}}\\
	&= \paren{\frac{m}{n}\cdot \pr{H \neq \emptyset}} \cdot \frac12 \cdot  \pr{O=1 \mid \set{I \in H}\land \set{B = X_I}}\\
	&= \frac{m}{2n} \cdot \pr{H \neq \emptyset} \cdot\pr{\Dist^{\cD,f}(J, Z, T) = 1 \mid I \in H}\\
	&> \frac{m}{2n}  \cdot \pr{H \neq \emptyset} \cdot (e^{\eps} p + 1/m)\\
	&= e^{\eps} \cdot \paren{\frac{m}{2n}  \cdot \pr{H \neq \emptyset} \cdot p} + \frac{1}{2n}\cdot \pr{H \neq \emptyset},\\
	&>  e^{\eps} \cdot \eex{i \la [n]}{\pr{X^*_i = -X_i}} + \delta,
\end{align*}
which concludes the proof. The last inequality holds since $\pr{H \neq \emptyset} \geq 1-\frac1{10000}$ and $\delta \leq \frac1{3n}$.

\end{proof}

\printbibliography

\iffull

\appendix

\section{\cite{HaitnerMST22}'s Protocol Cannot Imply $\OT$}\label{appendix:HaitnerMST22}

In this section, we show that for some carefully chosen $\CDP$ and accurate protocol $\Pi$, the joint view of the parties in \cite{HaitnerMST22}'s protocol (\cref{protocol:HaitnerMST22}) can be \emph{simulated} using a trivial protocol, without using $\Pi$ at all.

\begin{protocol}[\cite{HaitnerMST22}'s protocol]\label{protocol:HaitnerMST22}
    \item Oracle: An accurate $\CDP$ protocol $\Pi$ for the inner-product.
    \item Operation:~
    \begin{enumerate}
        \item $\Ac$ and $\Bc$ choose random inputs $x\in \set{-1,1}^n$ and $y\in \set{-1,1}^n$, respectively.
        \item The parties interact using $\Pi$ to get approximation $z$ of $\langle x,y \rangle$.
        \item $\Ac$ chooses a random string $r \gets \zn$, and sends $r, x_r=\set{x_i \colon r_i=1}$ to $\Bc$. $\Bc$ replies with $y_{- r}=\set{y_i \colon r_i=0}$.
        \item Finally, $\Ac$ computes and outputs $\out_\Ac=\langle x_{-r},y_{-r} \rangle$, and $\Bc$ outputs $\out_\Bc=z-\langle x_{r},y_{r} \rangle$.
\end{enumerate}
\end{protocol}

To see this, assume that $\Pi$ is a protocol that on inputs $x$ and $y$, outputs $z=\langle x,y \rangle+e_\Ac+e_\Bc$, where $e_\Ac$ and $e_\Bc$ are independent samples from the $\Lap(2/\eps)$ distribution. Moreover, assume that $\Pi$ reveals $e_\Ac$ to $\Ac$ and $e_\Bc$ to $\Bc$ (and nothing else).  Such a protocol is indeed differential private, and it can be implemented using secure multi-party computation. Moreover, as we show in this work, such a protocol can be used to construct OT. However, when composed with the $\KA$ protocol of \cite{HaitnerMST22}, the resulting protocol can be simulated trivially.\footnote{More formally, and using the definition given in \cref{sec:protocol}, we claim that the channel induced by executing \cref{protocol:HaitnerMST22} with oracle access to the channel $$\Pi=\set{((x,(\langle x,y \rangle + e_A+e_B,e_A)),(y,(\langle x,y \rangle + e_A+e_B,e_B)))}_{x,y\gets \oo^n, e_A,e_B\gets \Lap(2/\epsilon)}$$ is a trivial channel.} 

Indeed, note that the view of $\Ac$ in $\Pi$ only contains $x,z$ and  $e_A$, while the view of $\Bc$ only contains $y,z$ and $e_B$. In this case, the view of $\Ac$ in the $\KA$ protocol of \cite{HaitnerMST22} contains $x,z, e_A, r$ and $y_{-r}$, while the view of $\Bc$ contains $y,z, e_B, r$ and $y_{r}$.
We next explain how to simulate this view without using $\Pi$. Consider the following protocol $\Pi'$ that  simulates the $\KA$ protocol in a reverse order:

\begin{protocol}[The simulation $\Pi'$]\label{protocol:trivial}
    \item Operation:~
    \begin{enumerate}
        \item $\Ac$ and $\Bc$ choose random inputs $x\in \set{-1,1}^n$ and $y\in \set{-1,1}^n$, respectively.
               \item $\Ac$ chooses a random string $r \gets \zn$, and sends $r, x_r=\set{x_i \colon r_i=1}$ to $\Bc$. $\Bc$ replies with $y_{- r}=\set{x_i \colon r_i=0}$.
               \item $\Ac$ samples $e_A\gets \Lap(2/\eps)$ and sends $z_A=\langle x_{r},y_r\rangle + e_A$  to $\Bc$. $\Bc$ samples $e_B\gets \Lap(2/\eps)$ and sends $z_B=\langle x_{-r},y_{-r}\rangle + e_B$  to $\Ac$.
        \item The output of the protocol is $z=z_A+z_B$.
\end{enumerate}
\end{protocol}
Clearly, \cref{protocol:trivial} is a trivial protocol and does not use any cryptographic assumptions. However, the views of $\Ac$ and $\Bc$ in $\Pi'$ contain all the information learned by the parties in the $\KA$ protocol ($(x,z,e_A,r,y_{-r})$ and $(y,z,e_B,r,x_{r})$ respectively). Moreover, we claim that the parties in $\Pi'$ do not learn any information that the parties in the $\KA$ protocol did not learn. Indeed, the only new value learned by $\Bc$, $z_A,$ can be also computed by $\Bc$ in the $\KA$ protocol by computing $z-e_B-\langle x_r,y_r \rangle$. Similarly, $z_B$ can be already computed by $\Ac$. 

We note that every protocol that uses only a communication channel cannot be used to construct $\OT$ unless $\OT$ already exists, and similarly, every protocol that uses  $\Pi'$ as a subroutine (in a black-box manner) cannot be used to construct $\OT$. Since \cref{protocol:HaitnerMST22} (that is, the views and outputs of the parties when running \cref{protocol:HaitnerMST22}) are the same as $\Pi'$, we conclude that \cref{protocol:HaitnerMST22} could not be used to construct $\OT$ in a black-box manner.  


\section{\WEC from \AWEC (Proof of \cref{lemma:AWEC-to-WEC})}\label{sec:AWEC-to-WEC}

In this section, we prove \cref{lemma:AWEC-to-WEC} and  show how to implement \WEC (\cref{def:WEC}) from \AWEC (\cref{def:AWEC}) using a \ppt protocol. Crucially, the security proof is constructive, so that it could be used in the computational case as well (see, \cref{cor:CompAWEC to CompWEC}).


To prove \cref{lemma:AWEC-to-WEC}, we will need the following easy version of Goldriech-Levin \cite{GoldreichL89}.
\begin{lemma}
\label{prel:gl:weak:prob}
There exists a \ppt oracle-aided  algorithm $\Dec$ such that the following holds. Let $n\in N$ be a number, $x\in \zn$, and 
 and let $\Pred$ be an algorithm such that
\begin{align*}
\ppr{ r\gets \zn}{\Pred(r)=\GL(x,r)} > 3/4+0.001,
\end{align*}
 where $\GL(x,r)\eqdef \langle x,r \rangle$ is the Goldreich-Levin predicate. 
Then $\pr{\Dec^\Pred(1^n)=x}=1-\negl(n)$.
\end{lemma}
\begin{proof}[Proof of \cref{prel:gl:weak:prob}]
We use $\Pred$ to decode each bit of $x$ separately. For every $i$, let $e_i$ be the vector that has $1$ in the $i$-th entry, and $0$'s everywhere else. Observe that, for a uniformly chosen $R\gets \zn$, 
$$\pr{\Pred(R)=\GL(x,R) \land \Pred(R\xor e_i)=\GL(x,R\xor e_i)}\ge 1/2+0.001.$$
Thus,
$$\pr{\Pred(R)\xor\Pred(R\xor e_i) =\GL(x,R) \xor \GL(x,R\xor e_i)}\ge 1/2+0.001.$$
By linearity of the inner product we get that,
$$\pr{\Pred(R)\xor\Pred(R\xor e_i) =x_i}\ge 1/2+0.001.$$
Let $\Dec$ be the algorithm that for every $i$, computes $\Pred(R)\xor\Pred(R\xor e_i)$  for $n$ random values of $R$, and let $x'_i$ to be the majority of the outputs. Then, $\Dec$ outputs $x'=x'_1,\dots,x'_n$. By Chernoff bound, $x'_i$ is equal to $x_i$ with all but negligible probability. By the union bound, the above is true for all $i$'s simultaneously  with all but negligible probability, as we wanted to show.
\end{proof}
We are now ready to prove \cref{lemma:AWEC-to-WEC}.
\begin{proof}[Proof of \cref{lemma:AWEC-to-WEC}]
We now define the protocol $\Lambda^C$ as follows:
\begin{protocol}[$\Lambda^C=(\Ac,\Bc)$]
\item Oracle access: A channel $C =((\OA,\VA),(\OB,\VB))$.
	\item Operation:
	\begin{enumerate}
            \item Sample $((\oA,\vA),(\oB,\vB))\from C$. $\Ac$ gets $(\oA,\vA)$ and $\Bc$ gets $(\oB,\vB)$. 
		
			\item $\Ac$ chooses a random offset $s\gets [1000\ell]$ and sends it to $\Bc$. Let $\oA'=\ceil{\frac{\oA+s}{1000\ell}}$ and $\oB'=\ceil{\frac{\oB+s}{1000\ell}}$ (if $\oB=\bot$, let $\oB'=\bot$). 
        \item $\Ac$ chooses a random vector $r\from \zo^{\log(n)}$ and sends it to $\Bc$. Let $\hoA=\langle \oA',r \rangle$ and $\hoB=\langle \oB',r \rangle$ (if $\oB'=\bot$, let $\hoB=\bot$).
        \item $\Ac$ outputs $\hoA$ and $\Bc$ outputs $\hoB$.
        \end{enumerate}
\end{protocol}
Let $\tilde{C}$ be the channel induces by $\Lambda^C=(\Ac,\Bc)$ defined above. 
Let $\OA,\VA,\OB,\VB,S,R,\OA',\OB',\hOA,\hOB$ be the random variables that corresponds to the value of $\oA,\vA,\oB,\vB,s,r,\oA',\oB',\hoA,\hoB$ in a random execution of the above protocol. Denote $\hVA=(\VA,S,R)$ and $\hVB=(\VB,S,R)$ and note that $(\hVA,\hOA),(\hVB,\hOB)$ defines the channels $\tilde{C}$.

We now prove that if $C=((\OA,\VA),(\OB,\VB))$ is an $(\ell,\alpha,p,q)$-\AWEC then $\tilde{C}$ is an $(\alpha'=\alpha+0.001,\: p' = p ,\:  q' = 1/2 + 2.001q)$-\WEC.

\paragraph{Agreement:} If $\size{\OA-\OB}\le \ell$, then $\ppr{S}{\ceil{\frac{\OA+S}{1000\ell}}\ne \ceil{\frac{\OB+s}{1000\ell}}}\le 1/1000$. Thus, 
\begin{align*}
    \pr{\hOA\ne \hOB\mid \hOB\ne \bot} &\leq\pr{|\OA - \OB|\leq \ell\mid \OB\ne \bot}+1/1000\\
    &= \eps+1/1000=\eps'
\end{align*}

\paragraph{$\Bc$'s privacy:} Recall that the view of $\Ac$ in the above protocol is $\hVA=(\VA,S,R)$, and note that $S,R$ are independent of $\OB$. Since, $\hOB=\bot$ iff $\OB=\bot$ it follows that it follows that for every algorithm $\Dc$:
    \begin{align*}
    &\size{{\pr{\Dc(\hVA) = 1 \mid \hOB \neq \bot} - \pr{\Dc(\hVA) = 1 \mid \hOB = \bot}} }\\
    &=\size{{\pr{\Dc(\VA,S,R) = 1 \mid \OB \neq \bot} - \pr{\Dc(\VA,S,R) = 1 \mid \OB = \bot}} }\\
    &=\size{\pr{\Dc(\VA) = 1 \mid \OB \neq \bot} - \pr{\Dc(\VA) = 1 \mid \OB = \bot}} \le p=p'.
    \end{align*}
\paragraph{$\Ac$’s privacy:}  Assume towards a contradiction that there exists an algorithm $\Dc$ such that

\begin{align}\label{eq:avrging}
\pr{\Dc(\hVB)=\hOA  \mid \hOB\neq\bot} \ge\frac{1+q'}{2}=  3/4+q+0.01.
\end{align}
Let $\cG=\set{(\vB,s)\colon \ppr{R}{\Dc(\vB,s,R)=\hOA\mid\VB=\vB,S=s,\hOB\neq \bot}\ge 3/4+0.001}$, and first note that $\ppr{\VB,S}{(\VB,S)\in \cG\mid \hOB\neq \bot}\ge q+0.009$.
Indeed, otherwise it holds that
\begin{align*}
&\pr{\Dc(\hVB)=\hOA  \mid \hOB\neq\bot}\\
& = \ppr{\VB,S}{(\VB,S)\in \cG\mid \hOB\neq \bot} + \ppr{\VB,S}{(\VB,S)\notin \cG\mid \hOB\neq \bot}\cdot(3/4+ 0.001)\\
&< (q+0.009) + 1\cdot (3/4+ 0.001)\\
&= 3/4+ q+0.01
\end{align*}
in contradiction to \cref{eq:avrging}.
    Next, by \cref{prel:gl:weak:prob} there exists an algorithm $\Dc'$ such that 
    $$
    \pr{\Dc'(\VB,S)=O'_A \mid (\VB,S)\in\cG, \OB'\neq \bot}\ge 1-o(1)
    $$ 
    Which implies that, 
    \begin{align*}
    \pr{\Dc'(\VB,S)=O'_A \mid  \OB'\neq \bot}
    &\ge \pr{\Dc'(\VB,S)=O'_A \mid (\VB,S)\in\cG, \OB'\neq \bot}\cdot \pr{\Dc'(\VB,S)\in\cG \mid\hOB\neq \bot}\\
    &\ge (q+0.009)(1-o(1))\\
    &\ge q.
    \end{align*}
    Since by definition $\size{(O'_A\cdot 1000\ell-S)-\OA}\le 1000\ell$, and $S$ is independent of $\VB$, it follows that there exists an algorithm $\Dc''$ such that  
    $$\pr{\size{\Dc''(\VB)-\OA}\le 1000\ell\mid \OB\neq \bot}\ge \delta(1-2\alpha)> q.$$
Contradicting the fact that $C$ is an $(\ell,\alpha,p,q)$-\AWEC.
\end{proof}
\section{Missing Proofs}\label{sec:missing-proofs}

\subsection{Proving \cref{prop:hard-to-guess-inf,prop:hard-to-guess-comp}}\label{sec:missing-proofs:hard-to-guess}

We make use of the following claim.

\begin{claim}\label{claim:X-star}
    Let $X \la \oo$ and let $X^*$ be a random variable over $\set{-1,1,\bot}$ (correlated with $X$) such that for every $b,b' \in \oo$:
    \begin{align*}
        \pr{X^* = b \mid X = b'} \leq e^{\eps}\cdot \pr{X^* = b \mid X = -b'} + \delta.
    \end{align*}
    Then
    \begin{align*}
        \pr{X^* = X} \leq e^{\eps}\cdot \pr{X^* = -X} + \delta.
    \end{align*}
\end{claim}
\begin{proof}
    Compute
    \begin{align*}
        \pr{X^* = X } 
        &= \frac12 \cdot \pr{X^* = -1 \mid X = -1} + \frac12 \cdot \pr{X^* = 1 \mid X = 1}\\
        &\leq  \frac12 \cdot \paren{e^{\eps}\cdot \pr{X^* = -1 \mid X = 1} + \delta} + \frac12 \cdot \paren{e^{\eps}\cdot \pr{X^* = 1 \mid X = -1} + \delta}\\
        &= e^{\eps} \cdot \paren{\frac12 \cdot \pr{X^* = -1 \mid X = 1} + \frac12 \cdot \pr{X^* = 1 \mid X = -1}} + \delta\\
        &= e^{\eps} \cdot \pr{X^*_i = -X_i} + \delta.
    \end{align*}
\end{proof}

We next prove \cref{prop:hard-to-guess-inf}, restated below.

\begin{proposition}[Restatement of \cref{prop:hard-to-guess-inf}]
    \propHardToGuessInf
\end{proposition}
\begin{proof}
    
    Fix $b,b' \in \oo$ and $i \in [n]$. By \cref{claim:X-star}, it is sufficient to prove that
    \begin{align}\label{eq:X_i^*-goal}
        \pr{X_i^* = b \mid X_i = b'} \leq e^{\eps}\cdot \pr{X_i^* = b \mid X_i = -b'} + \delta.
    \end{align}
    For $x_{-i} = (x_1,\ldots,x_{i-1},x_{i+1},\ldots,x_n) \in \oo^{n-1}$, define the function
    \begin{align*}
        h_{x_{-i}}(y) = g(i,x_{-i},y).
    \end{align*}
    Since $f$ is $(\eps,\delta)$-\DP, for any $x_{-i} \in \oo^{n-1}$ it holds that
    \begin{align*}
        \pr{h_{x_{-i}}(f(x_1,\ldots,x_{i-1}, b', x_{i+1},\ldots,x_n)) = b} \leq e^{\eps}\cdot \pr{h_{x_{-i}}(f(x_1,\ldots,x_{i-1}, -b', x_{i+1},\ldots,x_n)) = b} + \delta.
    \end{align*}
    Thus,
    \begin{align*}
        \pr{X_i^* = b \mid X_i = b'}
        &= \eex{x_{-i} \la \oo^{n-1}}{\pr{g(i,x_{-i}, f(x_1,\ldots,x_{i-1}, b', x_{i+1},\ldots,x_n)) = b }}\\
        &= \eex{x_{-i} \la \oo^{n-1}}{\pr{h_{x_{-i}}(f(x_1,\ldots,x_{i-1}, b', x_{i+1},\ldots,x_n)) = b}}\\
        &\leq e^{\eps}\cdot \eex{x_{-i} \la \oo^{n-1}}{\pr{h_{x_{-i}}(f(x_1,\ldots,x_{i-1}, -b', x_{i+1},\ldots,x_n)) = b}} + \delta\\
        &= e^{\eps}\cdot\pr{X_i^* = b \mid X_i = -b'}+ \delta.
    \end{align*}
\end{proof}

We next prove \cref{prop:hard-to-guess-comp}, restated below.

\begin{proposition}[Restatement of \cref{prop:hard-to-guess-comp}]
    \propHardToGuessComp
\end{proposition}
\begin{proof}
    
    In the following, fix $b,b' \in \oo$, and for $\pk \in \bbN$, $i \in [n(\pk)]$ and $x_{-i} \in \oo^{n(\pk)-1}$, define 
    \begin{align*}
        h_{\pk}^{i,x_{-i}}(y) = b\cdot g_{\pk}(i,x_{-i},y).
    \end{align*}
    Note that for any ensemble $\set{(i,x_{-i})_{\pk} \in [n(\pk)]\times \oo^{n(\pk)-1}}_{\pk \in \bbN}$, the circuit family \\$\set{h_{\pk}^{(i, x_{-i})_{\pk}} = g_{\pk}((i,x_{-i})_{\pk},\cdot)}_{\pk \in \bbN}$ has poly-size. 
    Since $f = \set{f_{\pk}}_{\pk \in \bbN}$ is $(\eps,\delta)$-\CDP, then for large enough $\pk$, the following holds for every $i \in [n(\pk)]$ and $x_{-i} \in \oo^{n(\pk)-1}$:
    \begin{align*}
        \lefteqn{\pr{h_{\pk}^{i, x_{-i}}(f_{\pk}(x_1,\ldots,x_{i-1}, b', x_{i+1},\ldots,x_n)) = 1}}\\
        &\leq e^{\eps(\pk)}\cdot \pr{h_{\pk}^{i,x_{-i}}(f(x_1,\ldots,x_{i-1}, -b', x_{i+1},\ldots,x_n)) = 1} + \delta(\pk),
    \end{align*}
    as otherwise, there would exist an ensemble $\set{(i,x_{-i})_{\pk} \in [n(\pk)]\times \oo^{n(\pk)-1}}_{\pk \in \cS}$ for an infinite set $\cS \subseteq \bbN$ such that the circuit family 
    $\set{h_{\pk}^{(i, x_{-i})_{\pk}}}_{\pk \in \cS}$ violates the $(\eps,\delta)$-\CDP property of $f$. 

    In the following, fix such large enough $\pk$ and $i \in [n]$ for $n = n(\kappa)$, let $X = (X_1,\ldots,X_{n}) \la \oo^{n}$ and $X_i^* = g_{\pk}(i,X_{-i},f_{\pk}(X_1,\ldots,X_n))$, and compute
    \begin{align*}
        \lefteqn{\pr{X_i^* = b \mid X_i = b'}}\\
        &= \eex{x_{-i} \la \oo^{n-1}}{\pr{g_{\pk}(i,x_{-i}, f_{\pk}(x_1,\ldots,x_{i-1}, b', x_{i+1},\ldots,x_n)) = b }}\\
        &= \eex{x_{-i} \la \oo^{n-1}}{\pr{h_{\pk}^{i,x_{-i}}(f(x_1,\ldots,x_{i-1}, b', x_{i+1},\ldots,x_n)) = 1}}\\
        &\leq e^{\eps(\pk)}\cdot \eex{x_{-i} \la \oo^{n-1}}{\pr{h_{\pk}^{i,x_{-i}}(f(x_1,\ldots,x_{i-1}, -b', x_{i+1},\ldots,x_n)) = 1}} + \delta(\pk)\\
        &= e^{\eps(\pk)}\cdot \eex{x_{-i} \la \oo^{n-1}}{\pr{g_{\pk}(i,x_{-i}, f_{\pk}(x_1,\ldots,x_{i-1}, -b', x_{i+1},\ldots,x_n)) = b}} + \delta(\pk)\\
        &= e^{\eps(\pk)}\cdot\pr{X_i^* = b \mid X_i = -b'}+\delta(\pk).
    \end{align*}
    Since the above holds for any $b, b' \in \oo$, we conclude by \cref{claim:X-star} that 
    \begin{align*}
        \pr{X_i^* = X_i} \leq e^{\eps(\pk)}\cdot \pr{X_i^* = -X_i} + \delta(\pk),
    \end{align*}
    as required.
\end{proof}

\subsection{Proving \cref{lemma:property1:prediction}}\label{sec:prediction-lemma}

To prove \cref{lemma:property1:prediction}, we use the following lemma that measures the distance between two uniform stings conditioned on a random index $i$ either being fixed to $0$ or to $1$.

\begin{lemma}\label{lem:distance-I}
    Let $R \la \zo^n$. For any (randomized) function $F:\{0,1\}^n\rightarrow \{0,1\}$ and $\alpha > 0$, it holds that
    \begin{align}\label{eq:f-alpha}
        \ppr{i \la [n]}{\size{\:\ex{F(R) \mid R_i = 0}-\ex{F(R) \mid R_i = 1}\:}\geq \alpha} \leq \frac{2}{n \alpha^2},
    \end{align}
    where the expectations are taken over $R$ and the randomness of $F$.
\end{lemma}

The proof of \cref{lem:distance-I} uses the following fact.

\begin{fact}[Proposition 3.28 in \cite{HaitnerMST22}]\label{fact:I}
	Let $R$ be uniform random variable over $\{0,1\}^n$, and let $I$ be uniform random variable over $\mathcal{I}\subseteq[n]$, independent of $R$, then $SD(R|_{R_I=0},R|_{R_I=1})\leq1/\sqrt{\size{\cI}}$.
\end{fact}

We first prove \cref{lem:distance-I} using \cref{fact:I}.

\begin{proof}[Proof of \cref{lem:distance-I}]
    Assume towards a contradiction that there exist $F$ and $\alpha$ such that \cref{eq:f-alpha} does not hold.
    Namely, for $m = 1/\alpha^2$, there exist more than $2m$ indices $i \in [n]$ with $$\size{\:\ex{F(R) \mid R_i = 0}-\ex{F(R) \mid R_i = 1}\:}\geq 1/\sqrt{m}.$$
    This implies that there exist $b \in \oo$ and more than $m$ indices $i \in [n]$ with $$\ex{F(R) \mid R_i = b}-\ex{F(R) \mid R_i = 1-b} \geq 1/\sqrt{m}$$ (denote this set by $\cI$). 
    Thus, we deduce for $I \la \cI$ that
    \begin{align}\label{eq:big-R_I-gap}
        \size{\ex{F(R) \mid R_I = 0}-\ex{F(R) \mid R_I = 1}} 
        &\geq \ex{F(R) \mid R_I = b}-\ex{F(R) \mid R_I = 1-b}\\
        &\geq 1/\sqrt{m}.\nonumber
    \end{align}
    On the other hand, note that
    \begin{align*}
        SD(F(R)|_{R_I=0},F(R)|_{R_I=1})
        \leq SD(R|_{R_I=0},R|_{R_I=1})
        \leq 1/\sqrt{\size{\cI}}
        < 1/\sqrt{m},
    \end{align*}
    where the second inequality holds by \cref{fact:I}.
    Thus, we conclude that
    \begin{align*}
        \size{\ex{F(R) \mid R_I = 0}-\ex{F(R) \mid R_I = 1}}
        &\leq SD(F(R)|_{R_I=0},F(R)|_{R_I=1}) \cdot \sup_{r \in \zo^n, s \in \zo^*}(\size{F_s(r)})\\
        &< 1/\sqrt{m},
    \end{align*}
    where $F_s(r)$ denotes the function $F$ when fixing its random coins to $s$. 
    The above inequality contradicts \cref{eq:big-R_I-gap}, concluding the proof of the lemma.
\end{proof}

Using \cref{lem:distance-I}, we now prove \cref{lemma:property1:prediction}, restated below.

\begin{lemma}[Restatement of \cref{lemma:property1:prediction}]
    \PredictionLemma
\end{lemma}
\begin{proof}
	
	Let $\gamma \in (0,1)$ and $n \in \bbN$ and consider the following algorithm $\Gc = \Gc_{\gamma}$:
	\begin{algorithm}
		\item Inputs: $i \in [n]$,  $z_{-i} = (z_1,\ldots,z_{i-1},z_{i+1},\ldots,z_n) \in \oo^{n-1}$ and $w \in \zo^*$.
		\item Parameter: $\gamma \in (0,1)$.
		\item Oracle: $F \colon \zo^n \times \oo^{\leq n} \times \zo^* \rightarrow \zo$.
		\item Operation:~
		\begin{enumerate}
			\item For $b \in \oo$:
			\begin{enumerate}
				\item Let $z^b = (z_1,\ldots,z_{i-1},b,z_{i+1},\ldots,z_n)$.
				\item Estimate $\mu^b \eqdef \eex{r \la \zo^n \mid r_i = 1, \: F}{F(r,z^b_{r}, w)}$ as follows:
				\begin{itemize}
					\item Sample $r_1,\ldots,r_{s} \la \set{r \in \zo^n \colon r_i = 1}$, for $s = \ceil{\frac{128 \log (12n)}{\gamma^2}}$, and then sample $\tilde{\mu}^b \la \frac1{s} \sum_{j=1}^s F(r_j,z^b_{r_j}, w)$ (using $s$ oracle calls to $F$).
				\end{itemize}
			\end{enumerate}
			\item Estimate $\mu^* \eqdef \eex{r \la \zo^n \mid r_i = 0, \: f}{f(r,z^1_{r}, w)}$ as follows:
			\begin{itemize}
				\item Sample $r_1,\ldots,r_{s}  \la \set{r \in \zo^n \colon r_i = 0}$, for $s = \ceil{\frac{128 \log (12n)}{\gamma^2}}$, and then sample $\tilde{\mu}^* \la \frac1{s} \sum_{j=1}^s F(r_j,z^1_{r_j}, w)$ (using $s$ oracle calls to $F$).
			\end{itemize}
			\item If exists $b \in \oo$ s.t. $\size{\tilde{\mu}^b - \tilde{\mu}^*} < \gamma/4$ and $\size{\tilde{\mu}^{-b} - \tilde{\mu}^*} > \gamma/4$, output $b$.
			\item Otherwise, output $\bot$.
		\end{enumerate}
	\end{algorithm}
	In the following, fix a pair of (jointly distributed) random variables $(Z,W) \in \oo^n \times \zo^*$ and a randomized function  $F \colon \zo^n \times \oo^{\leq n} \times \zo^* \rightarrow \oo$ that satisfy 
	\begin{align*}
		\size{\ex{F(R,Z_{R},W) - F(R,Z^{(I)}_{R},W)}} \geq \gamma,
	\end{align*}
	for $R \la \zo^n$ and $I \la [n]$ that are sampled independently. 
	Our goal is to prove that 
	
	\begin{align}\label{eq:predic-lemma:UB}
		\pr{\Gc^F(I,Z_{-I}, W) = -Z_I} \leq O\paren{\frac{1}{\gamma^2 n}},
	\end{align}
	and 
	\begin{align}\label{eq:predic-lemma:LB}
		\pr{\Gc^F(I,Z_{-I}, W) = Z_I} \geq \Omega(\gamma) - O\paren{\frac{1}{\gamma^2 n}}.
	\end{align}

	Note that
	\begin{align}\label{eq:D}
		\text{For every random variable }D \in [-1,1]\text{ with }\size{\ex{D}} \geq \gamma >0: \quad \pr{\size{D} > \gamma/2} > \gamma/2,
	\end{align}
	as otherwise, $\size{\ex{D}} \leq \ex{\size{D}} \leq 1\cdot \frac{\gamma}{2} + \frac{\gamma}{2}(1-\frac{\gamma}{2}) < \gamma$.

	By applying \cref{eq:D} with $D = \eex{r \la \zo^n, F}{F(r,Z_{r},W) - f(r,Z^{(I)}_{r},W)}$, it holds that
	\begin{align}\label{eq:main-lemma:good}
		\ppr{(i,z,w) \la (I,Z,W)}{\size{\eex{r \la \zo^n, F}{F(r,z_{r},w) - F(r,z^{(i)}_{r},w)}} > \gamma/2} > \gamma/2.
	\end{align}
	On the other hand, for every fixing of $(z,w) \in \Supp(Z,W)$, we can apply \cref{lem:distance-I} with the function $F_{z,w}(r) = F(r,z_{r},w)$ and with $\alpha =\frac{\gamma}{16}$ to obtain that
	\begin{align*}
		\ppr{i\la I}{\: \size{\eex{r \la \zo^n \mid r_i = 0, \: F}{F(r,z_{r},w)} - \eex{r \la \zo^n \mid r_i = 1, \: F}{F(r,z_{r},w)} \:} \geq \frac{\gamma}{16}} \leq \frac{512}{n \gamma^2}.
	\end{align*}
	But since the above holds for every fixing of $(z,w)$, then in particular it holds that
	\begin{align}\label{eq:main-lemma:bad}
		\ppr{(i,z,w) \la (I,Z,W)}{\: \size{\eex{r \la \zo^n \mid r_i = 0, \: F}{F(r,z_{r},w)} - \eex{r \la \zo^n \mid r_i = 1, \: F}{F(r,z_{r},w)} \:} \geq \frac{\gamma}{16}} \leq \frac{512}{n \gamma^2}.
	\end{align}
	
	We next prove \cref{eq:predic-lemma:UB,eq:predic-lemma:LB} using \cref{eq:main-lemma:good,eq:main-lemma:bad}.
	
	In the following, for a triplet $t = (i,z,w) \in \Supp(I,Z,W)$, consider a random execution of $\Gc^F(i,z_{-i},w)$. For $x \in \set{-1,1,*}$, let $\mu^x_t$ be the value of $\mu^x$ in the execution, and let $M^x_t$ be the (random variable of the) value of $\tilde{\mu}^x$ in the execution. Note that by definition, it holds that $\mu_{i,z,w}^{z_i} = \eex{r \la \zo^n, F}{F(r,z_{r},w)}$ and $\mu_{i,z,w}^{-z_i} = \eex{r \la \zo^n, F}{F(r,z^{(i)}_{r},w)}$. Therefore, \cref{eq:main-lemma:good} is equivalent to 
	\begin{align}\label{eq:main-lemma:good2}
		\ppr{t=(i,z,w) \la (I,Z,W)}{\size{\mu_t^{z_i} - \mu_t^{-z_i}} > \gamma/2} > \gamma/2.
	\end{align}
	Furthermore, note that $\mu_{i,z,w}^* \eqdef \eex{r \la \zo^n \mid r_i = 0, \: F}{F(r,z^1_{r}, w)} = \eex{r \la \zo^n \mid r_i = 0, \: F}{F(r,z_{r}, w)}$ and that $\mu_{i,z,w}^{z_i} = \eex{r \la \zo^n \mid r_i = 1, \: F}{F(r,z^{z_i}_{r}, w)}$. Therefore, \cref{eq:main-lemma:bad} is equivalent to 
	\begin{align}\label{eq:main-lemma:bad2}
		\ppr{t = (i,z,w) \la (I,Z,W)}{\: \size{\mu_t^* - \mu_t^{z_i}}\geq \frac{\gamma}{16}}  \leq \frac{512}{n \gamma^2}.
	\end{align}
	
	We next prove the lemma using \cref{eq:main-lemma:good2,eq:main-lemma:bad2}.
	
	Note that by Hoeffding's inequality, for every $t = (i,z,w) \in \Supp(I,Z,W)$ and  $x \in \set{-1,1,*}$ it holds that $\pr{\size{M_t^x - \mu_t^x} \geq \frac{\gamma}{16}} \leq 2\cdot e^{-2 s \paren{\frac{\gamma}{16}}^2} \leq \frac1{6n}$, which yields that for every fixing of $t = (i,z,w) \in \Supp(I,Z,W)$, w.p.\ at least $1-\frac1{2n}$ we have for all $x \in \set{-1,1,*}$ that $\size{M_t^x - \mu_t^x} < \frac{\gamma}{16}$ (denote this event by $E_t$).
	
	The proof of \cref{eq:predic-lemma:UB} holds by the following calculation:
	\begin{align*}
		\lefteqn{\pr{\Gc^F(I,Z_{-I},W) = -Z_I}}\\
		&= \eex{(i,z,w) \la (I,Z,W)}{\pr{\Gc^F(i,z_{-i},w) = -z_i}}\\
		&= \eex{t = (i,z,w) \la (I,Z,W)}{\pr{\set{\size{M_t^* - M_t^{z_i}} > \gamma/4} \land \set{\size{M_t^* - M_t^{-z_i}} < \gamma/4}}}\\
		&\leq \eex{t =(i,z,w) \la (I,Z,W)}{\pr{\size{M_t^* - M_t^{z_i}} > \gamma/4}}\\
		&\leq \eex{t = (i,z,w) \la (I,Z,W)}{\pr{\size{M_t^* - M_t^{z_i}} > \gamma/4 \mid E_t}} + \frac{1}{2n}\\
		&\leq \ppr{t = (i,z,w) \la (I,Z,W)}{\size{\mu_t^* -  \mu_t^{z_i}} \geq \frac{\gamma}{16}} + \frac{1}{2n}\\
		&\leq \frac{512}{n \gamma^2} + \frac1{2n},
	\end{align*}
	The second inequality holds since $\pr{\neg E_t} \leq \frac1{2n}$ for every $t$. The penultimate inequality holds since conditioned on $E_t$, it holds that $\size{M_t^* - \mu_t^*} \leq \frac{\gamma}{16}$ and $\size{M_t^{z_i} - \mu_t^{z_i}} \leq \frac{\gamma}{16}$, which implies that $\size{M_t^* - M_t^{z_i}} > \gamma/4 \: \implies \: \size{\mu_t^* -  \mu_t^{z_i}} \geq \frac{\gamma}{4} - 2\cdot \frac{\gamma}{16} > \frac{\gamma}{16}$. The last inequality holds by \cref{eq:main-lemma:bad2}.
	
	It is left to prove \cref{eq:predic-lemma:LB}. 
	Observe that \cref{eq:main-lemma:good2,eq:main-lemma:bad2} imply with probability at least $\gamma/2 - \frac{512}{n \gamma^2}$ over $t = (i,z,w)\in (I,Z,W)$, we have $\size{\mu_t^* -  \mu_t^{z_i}} \leq  \frac{\gamma}{16}$ and $\size{\mu_t^{z_i}- \mu_t^{-z_i}} \geq\frac{\gamma}{2}$. Therefore, conditioned on event $E_t$ (i.e., $\forall x \in \set{-1,1,*}: \: \size{M_t^{x} - \mu_t^{x}} \leq \frac{\gamma}{16}$), we have for such triplets $t = (i,z,w)$ that 
	\begin{align*}
		\size{M_t^* - M_t^{-z_i}} 
		& \geq \size{\mu_t^{z_i} - \mu_t^{-z_i}} - \size{\mu_t^{z_i} - \mu_t^*} - \size{M_t^*  - \mu_t^*} -  \size{M_t^{-z_i} - \mu_t^{-z_i}}\\
		&\geq \frac{\gamma}{2} - 3\cdot \frac{\gamma}{16}\\
		&> \gamma/4,
	\end{align*}
	and 
	\begin{align*}
		\size{M_t^* - M_t^{z_i}}
		&\leq  \size{M_t^* - \mu_t^*} + \size{\mu_t^* - \mu_t^{z_i}} + \size{\mu_t^{z_i} - M_t^{z_i}} \\
		&\leq 3\cdot \frac{\gamma}{16}\\
		&< \gamma/4.
	\end{align*}
	
	Thus, we conclude that
	\begin{align*}
		\lefteqn{\pr{\Gc^F(I,Z_{-I},W) = Z_I}}\\
		&= \eex{(i,z,w) \la (I,Z,W)}{\pr{\Gc^F(i,z_{-i},w) = z_i}}\\
		&= \eex{t = (i,z,w) \la (I,Z,W)}{\pr{\set{\size{M_t^* - M_t^{-z_i}} > \gamma/4} \land \set{\size{M_t^* - M_t^{z_i}} < \gamma/4}}}\\
		&\geq \paren{1 - \frac1{2n}}\cdot \eex{t = (i,z,w) \la (I,Z,W)}{\pr{\set{\size{M_t^* - M_t^{-z_i}} > \gamma/4} \land \set{\size{M_t^* - M_t^{z_i}} < \gamma/4} \mid E_t}}\\
		&\geq \paren{1 - \frac1{2n}}\cdot \paren{\gamma/2 - \frac{512}{n \gamma^2}}\\
		&\geq \gamma/4 - \frac{512}{n \gamma^2},
	\end{align*}
	which proves \cref{eq:predic-lemma:LB}. The first inequality holds since $\pr{E_t} \geq 1-\frac1{2n}$ for every $t = (i,z,w)$, and the second one holds by the observation above.
	
\end{proof}

\fi

\end{document}